\documentclass[aps, prx, superscriptaddress, twocolumn,
notitlepage, footinbib, longbibliography,
floatfix]{revtex4-2}

\usepackage{tikz}
\usetikzlibrary{positioning}

\usepackage{graphicx}
\usepackage{amsmath,amssymb}
\usepackage{physics}
\usepackage{hyperref}
\usepackage{braket}
\usepackage{xcolor}
\usepackage{amsthm}
\usepackage{comment}

\usepackage{mathtools}

\newtheorem{theorem}{Theorem}
\newtheorem{proposition}[theorem]{Proposition}
\newtheorem{lemma}[theorem]{Lemma}
\newtheorem{remark}[theorem]{Remark}
\newtheorem{corollary}[theorem]{Corollary}

\usepackage{enumitem}

\begin{document}

\title{Extendibility of Fermionic Gaussian States}

\author{Amir-Reza Negari}
\affiliation{Perimeter Institute for Theoretical Physics, Ontario, Canada}

\author{Farzin Salek}
\affiliation{Perimeter Institute for Theoretical Physics, Ontario, Canada}
\affiliation{Dahlem Center for Complex Quantum Systems, Freie Universit\"{a}t Berlin, Berlin, Germany}
\affiliation{Institute for Quantum Computing (IQC), University of Waterloo, Ontario, Canada}

\begin{abstract}
We investigate $(k_1,k_2)$-extendibility of fermionic Gaussian states, a property central to quantum correlations and approximations of separability. We show that these states are $(k_1,k_2)$-extendible if and only if they admit a fermionic Gaussian extension, yielding a complete covariance-matrix characterization and a simple semidefinite program (SDP) whose size scales linearly with the number of modes. This provides necessary conditions for arbitrary fermionic states and is sufficient within the Gaussian setting. Our main result is a finite de Finetti--type theorem: we derive trace-norm bounds between $(k_1,k_2)$-extendible fermionic Gaussian states and separable states, improving previous exponential scaling to linear in the number of modes, with complementary relative entropy and squashed entanglement bounds. For two modes, upper and lower bounds match at order $1/\sqrt{k_1 k_2}$. Extendibility also provides operational support for one of the different notions of separability in fermionic systems. Finally, for fermionic Gaussian channels, we provide an SDP criterion for anti-degradability and show that entanglement-breaking channels coincide with replacement channels, implying no nontrivial entanglement-breaking fermionic Gaussian channels exist.
\end{abstract}
\maketitle
\textit{Introduction} --- 
Quantum entanglement underpins both foundational aspects of quantum theory and its technological applications~\cite{RevModPhys.81.865, Nielsen_Chuang_2010}. Beyond serving as a resource for quantum communication, computation, and metrology~\cite{Dowling2003Quantum,Hayden2006-cq}, its structure reveals fundamental constraints on how quantum correlations can be distributed~\cite{10.5555/2011706.2011707, RevModPhys.80.517, RevModPhys.82.277}. One of the most striking of these constraints is \emph{monogamy}~\cite{PhysRevA.61.052306,5388928,PhysRevLett.117.060501,PhysRevA.69.022309}: entanglement cannot be freely shared among multiple parties. This idea is formalized through \emph{extendibility} \cite{PhysRevLett.88.187904,PhysRevA.69.022308,Li2018-gb}: 
A bipartite state $\rho^{AB}$ is said to be \emph{$(k_{1},k_{2})$-extendible} if there exists a state $\widetilde{\rho}^{A_{1}\ldots A_{k_{1}} B_{1}\ldots B_{k_{2}}}$
such that 
(i) it is invariant under permutations of the $A$-systems and, separately, under permutations of the $B$-systems, and 
(ii) for every pair $(i,j)$, the two-body marginal satisfies $\widetilde{\rho}^{A_{i}B_{j}} = \rho^{AB}$.
Two-sided $(k_1,k_2)$-extendibility provides a symmetric characterization of entanglement shareability, capturing constraints on both subsystems and offering a more refined operational framework for analyzing bipartite correlations. 

A fundamental result is that infinite extendibility is equivalent to separability~\cite{1988LMaPh..15..255F, PhysRevA.71.032333, PhysRevLett.90.157903, Christandl2007-xq}.
 This setting naturally relates to finite quantum de Finetti theorems, which characterize permutation-invariant states as approximate convex combinations of independent and identically distributed (i.i.d.) product states~\cite{Renner2007-kp,Krumnow_2017,Christandl2007-xq}. Classical de Finetti results describe infinite extendibility, which implies exact product structure; finite versions, in contrast, quantify how close $(k_1,k_2)$-extendible states are to separable ones, providing operationally meaningful bounds. 
Although extendibility and de Finetti-type results have been extensively studied for systems with tensor-product Hilbert spaces and commuting observables—such as qubits and bosonic Gaussian states \cite{Lami_2019,doi:10.1142/S0129055X1750012X}—they remain far less explored in the fermionic setting \cite{Krumnow_2017,krumnow2024}, despite the ubiquity and fundamental importance of fermions across physics. 

Fermionic many-body states are foundational in condensed-matter physics and quantum chemistry, underpinning electronic structure and correlated phases~\cite{Langel2002-or,szabo1996modern}. They are likewise important in quantum information, where fermionic modes provide efficient models for simulation, benchmarking, and computation~\cite{bravyi2004LR,RevModPhys.82.277,Bravyi2017-xn,10.1063/1.5085428}. In the fermionic Gaussian (quasifree) regime these states arise as ground and thermal states of quadratic Hamiltonians, remain closed under quadratic dynamics, and admit efficient correlation-matrix simulation methods~\cite{Peschel_2009,IngoPeschel_2003}; their relevance is illustrated by paradigmatic topological models such as the Kitaev honeycomb lattice and the SSH chain~\cite{KITAEV20062,RevModPhys.60.781}. These systems further motivate work on fermionic resource theories, separability diagnostics, Gaussian channel structure \cite{PhysRevB.97.165123,PhysRevB.95.165101,Greplov__2018,PhysRevA.99.022310}.

For fermions, anticommutation relations and superselection rules alter the structure of “local” operations and the notion of entanglement. One can view entanglement either between indistinguishable particles (first quantization) \cite{PhysRevA.64.022303,PhysRevA.72.022302} or, more naturally for quantum information, between sets of modes assigned to different spatial regions (second quantization) \cite{Ba_uls_2007,PhysRevA.97.042325,PhysRevA.87.022338}. In this mode-based picture, Alice controls modes in region \(A\) and Bob those in region \(B\), and entanglement reflects correlations across these partitions. Unlike qubits, however, observables from different regions may anticommute, imposing algebraic constraints that make fermionic entanglement subtler. 

In this work we adopt the second-quantization framework and study the extendibility problem for fermionic Gaussian states.
We provide a complete characterization of two-sided extendibility for fermionic Gaussian states. Our analysis yields a linear-size semidefinite program, derived directly from their covariance matrices (CM), that gives necessary and sufficient conditions for extendibility. Our framework sharpens separability notions by establishing a consistent operational definition and, in doing so, identifies a distinctive class of \emph{classical--quantum binding states}—including fermionic Gaussian examples—where classical correlations necessarily imply quantum ones. We prove a tight de Finetti-type theorem bounding the trace-norm distance between $(k_1,k_2)$-extendible and separable states, derive upper bounds on both relative entropy and squashed entanglement, and characterize fermionic Gaussian channels via necessary and sufficient conditions for antidegradability, a sufficient condition for zero quantum capacity.

\emph{Fermionic Gaussian states:}  
We recall the basic framework for fermionic Gaussian states, following standard conventions~\cite{bravyi2004LR,Bravyi2017-xn,PhysRevB.100.245121,Greplov__2018}.  
An \(n\)-mode fermionic system is described by creation and annihilation operators \(\{\hat{c}_j^\dagger,\hat{c}_j\}_{j=1}^n\) obeying \(\{\hat{c}_j,\hat{c}_k^\dagger\}=\delta_{jk}\), \(\{\hat{c}_j,\hat{c}_k\}=0\). Introducing \(2n\) Hermitian Majorana operators \(\gamma_{2j-1}=\hat{c}_j+\hat{c}_j^\dagger\), \(\gamma_{2j}=i(\hat{c}_j-\hat{c}_j^\dagger)\), they satisfy \(\{\gamma_j,\gamma_k\}=2\delta_{jk}\) and generate the Clifford algebra.  
The second moments of a state \(\rho\) are encoded in its CM \(M\), defined by \(M_{qp}=\Tr(i\gamma_q\gamma_p\rho)\) for \(q\neq p\) (with \(M_{qq}=0\)). A valid CM is (i) real, (ii) antisymmetric, and (iii) satisfies the bona fide condition \(I+iM\ge0\). These conditions are necessary and sufficient for \(M\) to be a valid fermionic CM. Fermionic Gaussian states are precisely those fully specified by \(M\). In addition, all physical fermionic states obey the parity superselection rule $[\rho,P]=0$ \cite{Ba_uls_2007}, where
$P=i^{n}\prod_{j=1}^{n}\gamma_{2j-1}\gamma_{2j}$ and $P^{2}=I$.
Since $P\gamma_m P=-\gamma_m$ for every Majorana operator, all odd moments vanish; in particular,
$\Tr(\rho\,\gamma_m)=\Tr(P\rho P\,\gamma_m)=\Tr(\rho\,P\gamma_m P)=-\Tr(\rho\,\gamma_m)=0$.
Thus, the states we consider throughout are physical in this sense.

\textit{Extendibility of fermionic states:}
Let $\rho^{AB}$ be a (not necessarily Gaussian) bipartite fermionic state on $n=n_A+n_B$ modes. For any physical fermionic state the first moments vanish. Moreover, since Majorana operators are bounded and the Hilbert space is finite-dimensional for finite \(n\), all second moments are finite. The CM then takes the block form  
\begin{align}
\label{eq:main_CM}
M_{AB}=\begin{pmatrix} 
M_A & X \\ -X^{T} & M_B 
\end{pmatrix},
\end{align}
where \(M_A\) and \(M_B\) are the CMs of subsystems \(A\) and \(B\), respectively. The requirement that \(M_{AB}\) be real and antisymmetric determines the structure of the off-diagonal block: \(X\in\mathbb{R}^{2n_A\times 2n_B}\), which is otherwise arbitrary subject to the usual bona fide condition.

For simplicity we begin by considering $k$-extendibility on subsystem \(B\). Permutation invariance of \(\widetilde{\rho}^{AB_1\ldots B_k}\) constrains the CM to have the block form  
\begin{align}
\label{one-sided-CM}
\widetilde{M}_{AB_1\cdots B_k}=
\begin{pmatrix}
M_A & X & X & \cdots & X \\
-X^{T} & M_B & Y & \cdots & Y \\
-X^{T} & Y & M_B & \cdots & \vdots \\
\vdots & \vdots & \vdots & \ddots & Y \\
-X^{T} & Y & \cdots & Y & M_B
\end{pmatrix},
\end{align}
with \(X\) real and \(Y\) real antisymmetric. Establishing the latter is nontrivial; our argument departs from the bosonic case~\cite{Lami_2019}. As detailed in the SM~\cite{SupplementalMaterial}, it exploits the algebra of Majorana operators, and can be adopted for the bosonic setting via field–quadrature relations.
\nocite{trace-antonio,Audenaert_2007,bravyi2005classicalcapacityfermionicproduct,moriya2005}
 Since fermionic Gaussian states are fully determined by their CMs, any bona fide matrix of the form \eqref{one-sided-CM} represents a $k$-extendible fermionic Gaussian state. The analogous two-sided CM is given in the Supplemental Material (SM)~\cite{SupplementalMaterial}.
\begin{theorem}
\label{gaussian-extension}
A fermionic Gaussian state $\rho^{AB}_{\mathrm{G}}$ is $(k_1,k_2)$-extendible if and only if it admits a Gaussian $(k_1,k_2)$-extension.
\end{theorem}
\begin{proof}
Assume there exists a non-Gaussian extension $\widetilde{\rho}^{AB_1\ldots B_k}$. Its second moments must then satisfy the extension constraints of Eq.~\ref{one-sided-CM}. Since every admissible CM defines a unique fermionic Gaussian state, we may replace $\widetilde{\rho}^{AB_1\ldots B_k}$ by the Gaussian state $\widetilde{\rho}_{\text{G}}^{AB_1\ldots B_k}$ with the same CM. This state is therefore a valid Gaussian extension of $\rho^{AB}_{\mathrm{G}}$.
\end{proof}
An analogous conclusion was reached in the bosonic case using central-limit-theorem techniques~\cite{Lami_2019}, and these methods can in fact be adapted to the fermionic setting via a fermionic CLT~\cite{coffman2025measuringnongaussianmagicfermions}. Our approach instead uses a direct covariance-matrix argument, yielding a derivation that is more direct and concise without invoking central-limit theorems.

Focusing on the one-sided case, we give a necessary condition for $k$-extendibility on subsystem $B$ for general fermionic states; the full two-sided result is given in Theorem~\ref{theorem-extendible}.
Let $\rho^{AB}$ be a (possibly non-Gaussian) $k$-extendible state on $n_A+n_B$ modes with CM $M_{AB}$. Then there exists a CM $\Delta_B$ such that
\begin{align}
\label{one-sided-necessary}
   I^A \oplus \tfrac{1}{k} I^B + i M_{AB} 
   \;\;\ge\;\; \mathbf{0} \oplus i\!\left(1-\tfrac{1}{k}\right)\!\Delta_B .
\end{align}
If $\rho^{AB}$ is Gaussian, condition~\eqref{one-sided-necessary} is also sufficient.
We prove this using the block-matrix positivity criterion and its Schur complement~\cite[Thm.~1.12]{Horn2005}:
\begin{align}
\label{schur-complement}
  \begin{pmatrix} A & C \\ C^\dagger & B \end{pmatrix} \ge 0 
  \quad \Leftrightarrow \quad  
  A \ge 0,\quad B \ge C^\dagger A^{-1} C ,
\end{align}
with the singular case handled by standard methods \footnote{The case of non-invertible \(A\) can be treated by the standard \(\varepsilon\)-regularization technique~\cite{Lami_2019}; all intermediate reductions and the final result remain valid without assuming invertibility.}.

For a $k$-extendible state, the CM has the form~\eqref{one-sided-CM} and satisfies $I+i\widetilde{M}_{AB_1\cdots B_k}\ge0$. Introducing
\(
\ket{+} = \tfrac{1}{\sqrt{k}}\sum_{j=1}^k |j\rangle,
\) 
the condition becomes an instance of~\eqref{schur-complement} with
\begin{align*}
A &= I+iM_A, \,\,\, C= \sqrt{k}\,(\bra{+}_B \otimes iX), \\
B &= k\ketbra{+}\otimes iY + I \otimes (I+iM_B-iY).
\end{align*}
Since $A\ge0$, the Schur complement yields
\begin{multline*}
(I-\ketbra{+})\otimes (I+iM_B-iY) \\
+ \ketbra{+}\otimes \bigl(I+iM_B+(k-1)iY - k X^T (I+iM_A)^{-1} X\bigr) \ge 0.
\end{multline*}
Orthogonality implies each term is positive, giving
\begin{align*}
\frac{1}{k-1}\bigl[k X^T (I+iM_A)^{-1} X - I - iM_B\bigr] &\le iY\le I+iM_B.
\end{align*}
Defining $\Delta_B := M_B - Y$, the first inequality is equivalent to~\eqref{one-sided-necessary}, while the second ensures $\Delta_B$ is a valid CM.

We now state the full two-sided extension criterion, proved in SM~\cite{SupplementalMaterial}, where we also discuss the landscape of extendibility and show that extendibility from different sides is independent. In addition to providing necessary (and, for Gaussian states, sufficient) conditions, the theorem also enables the explicit construction of extensions as feasible points of the SDP.

\begin{theorem}
\label{theorem-extendible}
Let $\rho^{AB}$ be a $(k_1,k_2)$-extendible state (not necessarily Gaussian) with CM~\eqref{eq:main_CM}.
Then there exist real antisymmetric matrices $\Delta_A,\Delta_B$ such that 
$I+i\Delta_A\geq 0$ and $I+i\Delta_B\geq0$, and
\begin{align*}
I + i
\begin{pmatrix}
k_1 M_A - (k_1\!-\!1)\Delta_A & \sqrt{k_1 k_2}\,X \\[1mm]
-\sqrt{k_1 k_2}\,X^T & k_2 M_B - (k_2\!-\!1)\Delta_B
\end{pmatrix} \geq 0. \label{eq:global}
\end{align*}
If $\rho^{AB}$ is Gaussian, the above semidefinite feasibility condition is both necessary and sufficient.
\end{theorem}
\textit{Properties of extendible fermionic states.}—We now apply Theorem~\ref{theorem-extendible} to characterize quantum and classical correlations in extendible fermionic Gaussian states. Our first result places a universal constraint on cross–subsystem correlations that applies to \emph{all} fermionic states, not only Gaussians. Its proof is in SM \cite{SupplementalMaterial}.
\begin{lemma}
\label{k_1k_2-bound}
Let $\rho^{AB}$ be a $(k_1,k_2)$-extendible state with CM as in Eq.~\eqref{eq:main_CM}. The off-diagonal block $X$ of its CM satisfies
\begin{align*}
X^T X \;\leq\; \frac{I}{k_1 k_2}.
\end{align*}
\end{lemma}
This bound progressively suppresses two-point correlations between \(A\) and \(B\) as the extendibility parameters grow, pushing the CM toward block-diagonal form. In the limit of arbitrary extendibility from either side (\(k_1\!\to\!\infty\) or \(k_2\!\to\!\infty\)), all cross-correlations vanish and the CM factorizes, $M_{AB} = M_A \oplus M_B$,
thus providing a proof of the observation made in~\cite[Observation 1]{PhysRevA.97.042325}. From this we deduce:

\begin{proposition}
\label{k_1k_2-infty}
A bipartite fermionic Gaussian state \(\rho^{AB}_{\mathrm{G}}\) is arbitrarily extendible from either side if and only if it is a product state: $\rho^{AB}_{\mathrm{G}} = \rho^A_{\mathrm{G}} \otimes \rho^B_{\mathrm{G}}.$
\end{proposition}

\begin{proof}
(\(\Rightarrow\)) For Gaussian states, Wick’s theorem reduces all higher-order correlators to two-point functions. The bound above forces the off-diagonal block \(X\) to vanish in the infinite-extendibility limit, eliminating every cross-subsystem correlator and hence yielding \(\rho_{\text{G}}^{AB}=\rho_{\text{G}}^A\!\otimes\!\rho_{\text{G}}^B\).
(\(\Leftarrow\)) Conversely, any product Gaussian state is trivially \((k_1,k_2)\)-extendible for all \(k_1,k_2\) by tensoring the marginals. Thus both quantum and classical correlations vanish precisely in the arbitrarily extendible regime.
\end{proof}

This behavior contrasts sharply with bosonic or spin systems \cite{Christandl2007-xq}, where infinite extendibility guarantees separability but not necessarily a product structure. The difference stems from fermionic superselection rules, which restrict admissible decompositions of separable states.

Separable states are convex mixtures of product states, $\rho^{AB} = \sum_x p(x)\, \rho_x^A \otimes \rho_x^B$,
but superselection leads to inequivalent notions \cite{Ba_uls_2007}. Two classes are most relevant: 
\emph{(i)} in $\mathcal{S}2_{\pi}$ each component is physical,
$\forall x:\, [\rho_x^A \otimes \rho_x^B,P^{AB}]=0$, equivalently $[\rho_x^A,P^A]=[\rho_x^B,P^B]=0$;
\emph{(ii)} in $\mathcal{S}2'_{\pi}$ only the mixture is required to be physical, $[\rho^{AB},P^{AB}]=0$, even if the summands are not. Ref.~\cite{Ba_uls_2007} proves the strict inclusion $\mathcal{S}2_{\pi}\subset\mathcal{S}2'_{\pi}$. 

Guided by our extendibility bounds, we single out $\mathcal{S}2_{\pi}$ as the operationally natural separability notion for fermions. The key structural fact for Gaussians is the classical–quantum binding property, a term we introduce to denote the absence of purely classical correlations:
\begin{theorem}[Fermionic Gaussians are classical–quantum binding]
\label{no-classical}
For any fermionic Gaussian state $\rho^{AB}_{\text{G}}$,
\[
\rho_{\text{G}}^{AB} \in \mathcal{S}2_{\pi}\;\Longleftrightarrow\;\rho_{\text{G}}^{AB} = \rho_{\text{G}}^A \otimes \rho_{\text{G}}^B.
\]
\end{theorem}
The implication ``product \(\Rightarrow \mathcal{S}2_{\pi}\)'' is immediate. For the converse, let \(\rho_{\text{G}}^{AB}\in \mathcal{S}2_{\pi}\). By definition, it admits a convex decomposition \(\rho_{\text{G}}^{AB}=\sum_j p_j\,\rho_j^A\otimes \rho_j^B\), where each local state is parity preserving, i.e.\ \([P^A,\rho_j^A]=[P^B,\rho_j^B]=0\). Since Majorana operators are odd under parity, \(\{P^A,\gamma_m^A\}=\{P^B,\gamma_n^B\}=0\), it follows that \(\Tr(\gamma_m^A\rho_j^A)=0\) and \(\Tr(\gamma_n^B\rho_j^B)=0\) for every \(j,m,n\). Hence, for each term in the decomposition, \(\Tr(\gamma_m^A\gamma_n^B(\rho_j^A\otimes\rho_j^B))=\Tr(\gamma_m^A\rho_j^A)\Tr(\gamma_n^B\rho_j^B)=0\). Summing over \(j\), we obtain \(\Tr(\gamma_m^A\gamma_n^B\,\rho_{\text{G}}^{AB})=0\) for all \(m,n\). Therefore the cross-correlation block of the covariance matrix vanishes, so the covariance matrix is block diagonal. Since a fermionic Gaussian state is completely determined by its covariance matrix, this block-diagonal structure implies \(\rho_{\text{G}}^{AB}=\rho_{\text{G}}^A\otimes\rho_{\text{G}}^B\). Together with Proposition~\ref{k_1k_2-infty}, we conclude that the following statements are equivalent for $\rho^{AB}_{\text{G}}$: (i) $\rho_{\text{G}}^{AB}=\rho_{\text{G}}^A\otimes\rho_{\text{G}}^B$; (ii) $\rho_{\text{G}}^{AB}$ is arbitrarily extendible; and (iii) $\rho^{AB}_{\text{G}}\in \mathcal{S}2_{\pi}$.

Thus, for Gaussian states the separability hierarchy collapses: the $\mathcal{S}2_{\pi}$ notion coincides with product structure and with arbitrary extendibility. Under weaker notions such as $\mathcal{S}2'_{\pi}$ or $\mathcal{S}1_{\pi}$ \cite{Ba_uls_2007}, states may retain residual classical correlations, thereby failing to be extendible; this mismatch makes such notions operationally irrelevant in practice. Finally, while Proposition~\ref{k_1k_2-infty} is specific to Gaussians, any non-Gaussian state in $\mathcal{S}2_{\pi}$ still has vanishing cross second moments (each component is parity-preserving), whereas states in $\mathcal{S}2'_{\pi}$ can sustain nonzero cross second moments—underscoring $\mathcal{S}2_{\pi}$ as the operationally justified notion of fermionic separability.
In what follows, unless stated otherwise, we use ``separable'' to mean separable in the $\mathcal{S}2_{\pi}$ sense, the strongest notion of fermionic separability \cite{Ba_uls_2007}.

Extendible states provide a powerful approximation to the separable set and underlie efficient entanglement criteria~\cite{PhysRevLett.90.157903,PhysRevA.80.052306,10.1145/1993636.1993683}. Their importance stems from the quantum de Finetti theorem~\cite{Renner2007-kp}. For finite $k$, the central question is how well extendibility approximates separability. Bounding this gap offers an operational measure of the effectiveness of extendibility. Since separability testing and entanglement detection underpin many quantum protocols, our finite de Finetti bounds have direct implications for security analysis and resource estimation in such applications~\cite{Renner2007-kp,Christandl2007-xq}. We now establish finite de Finetti bounds for fermionic Gaussian states, focusing on the two-sided, $(k_1,k_2)$-extendible case. The following theorem applies uniformly to all separability notions.
\begin{theorem}[Finite de Finetti bound for fermionic Gaussian states]
\label{distance-theorem}
Let $\rho^{AB}_{G}$ be a $(k_1,k_2)$-extendible fermionic Gaussian state on $n_A+n_B$ modes. Define $T:=\frac{2}{\sqrt{k_1 k_2}}\min(n_A,n_B,\sqrt{k_1 k_2})$ and $h(x):=-x\log_2 x-(1-x)\log_2(1-x)$.
Then
\begin{align}
\text{(i)}\;& \| \rho^{AB}_{G} - \mathrm{SEP}(A\!:\!B) \|_1\le T, \\
\text{(ii)}\;& E_R(\rho^{AB}_{\text{G}})\le \tfrac12(n_A+n_B)T + h(T/2), \\
\text{(iii)}\;& E_{\mathrm{sq}}(\rho^{AB}_{\text{G}})\le \tfrac14(n_A+n_B)T + \tfrac12 h(T/2).
\end{align}
Here $\|\rho-\mathrm{SEP}\|_1 := \min_{\sigma \in \mathrm{SEP}(A\!:\!B)} \|\rho-\sigma\|_1$
 (trace distance to separable states), 
$E_R(\rho)$ (relative entropy of entanglement) $:=\inf_{\sigma\in\mathrm{SEP}} S(\rho\Vert\sigma)$ with $S(\rho\Vert\sigma)=\Tr[\rho(\log\rho-\log\sigma)]$, 
$E_{\mathrm{sq}}(\rho^{AB})$ (squashed entanglement) $:=\tfrac12\inf_{\rho^{ABE}} I(A\!:\!B|E)$ with $I(A\!:\!B|E)=S(AE)+S(BE)-S(E)-S(ABE)$, 
and $S(\tau)=-\Tr[\tau\log\tau]$ is the von Neumann entropy.
\end{theorem}

Proofs are given in the SM. These bounds differ sharply from the bosonic and spin cases. For qubits~\cite{Christandl2007-xq} and for general fermionic settings~\cite[Corollary~15]{krumnow2024}, the distance between extendible and separable states scales exponentially with the number of subsystems or modes. In contrast, our bounds scale linearly in trace norm and only polynomially for $E_R$ and $E_{\mathrm{sq}}$, replacing the previously exponential scaling~\cite{krumnow2024} with a linear one. The bosonic analogue~\cite[Theorem~3]{Lami_2019} exploits that any bosonic Gaussian state admits a separable counterpart with a CM proportional to its own~\cite{10.1063/1.3043788}, ensuring commutativity. No such reduction exists for fermions; instead, the key simplification in the fermionic Gaussian setting is that extendibility can be expressed directly as a linear matrix inequality for the covariance matrix, and Gaussianity ensures that controlling the off-diagonal covariance block \(X\) controls the full state. This mechanism is specific to quasifree states: in non-Gaussian or interacting fermionic systems, higher-order cumulants are independent degrees of freedom and are not fixed by the covariance matrix, so no analogous dimension-linear characterization is available. A striking consequence is that, while in the bosonic case the distance in the one-sided \(k\)-extendible setting asymptotically vanishes as \(1/k\), in the fermionic case it decays only as \(1/\sqrt{k}\). We establish the tightness of this scaling with a two-mode example, yielding a lower bound of \(1/\sqrt{k_1 k_2}\) and an upper bound of \(2/\sqrt{k_1 k_2}\). In the SM we further prove that the dependence on the number of modes is also optimal.

Although Theorem~\ref{distance-theorem} is stated as a worst-case bound and does not assume any geometric structure, its effective content can be substantially sharper in local settings.  If the Gaussian state is the ground state of a local gapped quadratic Hamiltonian, then two-point correlations typically decay exponentially with distance, so the off-diagonal covariance block \(X\) across a spatial bipartition is concentrated near the boundary between the two regions. Consequently, the relevant contribution to \(\|X\|_1\) comes only from modes within a correlation length of the cut, so the effective number of modes entering the bound is governed by the boundary size \(|\partial A|\sim n^{1-1/D}\), rather than the total number of bulk modes \(|A|\sim n\). We further study this picture numerically in the SM in a geometrically local setting, namely a disordered Kitaev chain, and find that increasing disorder makes the state more likely to be \((k_1,k_2)\)-extendible.

\emph{Extendibility of Fermionic Gaussian Channels:}
We now turn to fermionic Gaussian channels (FGCs), i.e., completely positive, trace-preserving maps that preserve Gaussianity. As shown in~\cite{bravyi2004LR}, an $m\!\to\! n$ fermionic Gaussian channel (FGC) 
$\mathcal{N}^{A\to B}$ is specified by a real $2n\times 2m$ matrix $X$ and a real antisymmetric $2n\times 2n$ matrix $N_B$ obeying $I + i N_B - X X^T \ge 0$, which implies $X X^T \le I$ and that $N_B$ is a valid CM. Its action on a CM is $M_A \mapsto X M_A X^{T} + N_B$. The Choi--Jamiołkowski state $C_{\mathcal{N}}^{AB}$, obtained by applying $\mathcal{N}$ to (the first) half of a maximally entangled fermionic Gaussian state, is itself Gaussian with CM 
\begin{align}
\label{fgc-cm}
    M_{\mathcal{N}}=\begin{pmatrix} N_B & X \\ -X^T & 0 \end{pmatrix}.
\end{align}

Determining when a quantum channel has zero capacity is a fundamental problem in quantum information. A prominent example is the class of antidegradable channels, where the environment can reconstruct the channel output; by the quantum no-cloning theorem, their quantum capacity vanishes. Characterizing antidegradability has thus attracted significant interest~\cite{10.1063/1.2953685,khanian2025strongconverseboundsprivate,PhysRevA.110.012460,Greplov__2018}.

Antidegradability is closely related to $k$-extendible channels: a channel $\mathcal{N}$ is $k$-extendible if its Choi state $C_{\mathcal{N}}^{AB}$ is $k$-extendible with respect to the output system $B$ \cite{PhysRevA.104.022401,PhysRevLett.123.070502}. In particular, antidegradability is equivalent to $2$-extendibility of the Choi state~\cite{PhysRevA.79.062307}. Applying Theorem~\ref{theorem-extendible} to the Choi CM $M_{\mathcal{N}}$ in Eq.~\eqref{fgc-cm} with $k_1=2$, $k_2=1$, $M_B=0$, and $M_A=N_B$ yields a necessary and sufficient condition for antidegradability of fermionic Gaussian channels:
\begin{proposition}
\label{prop:antideg}
A fermionic Gaussian channel $(X, N_B)$ is antidegradable if and only if there exists a real antisymmetric matrix $\Delta$ satisfying
\begin{align*}
   i\Delta \;\le\; I, \qquad i\Delta \;\le\; I + 2i N_B - 2X X^T \,.
\end{align*}
\end{proposition}

A simple example illustrating Proposition~\ref{prop:antideg} is the fermionic Gaussian attenuation channel.  
Take \(X=\sqrt{\lambda}\,\mathbb{I}\) and \(N_B=0\), for example on a two-mode system.  
The channel acts on covariance matrices as \(\Gamma_{\mathrm{out}}=\lambda\,\Gamma_{\mathrm{in}}\), and the
complete-positivity condition holds exactly for \(0\le \lambda \le 1\).  
For these parameters, the matrix appearing in Proposition~\ref{prop:antideg} is
\(I+2iN_B-2XX^{\mathsf T}=I-2\lambda I=(1-2\lambda)I\), from which one finds that the channel is
antidegradable precisely when \(0\le \lambda \le \tfrac12\).  
Thus sufficiently strong attenuation places the channel in the antidegradable regime, where its quantum
capacity is zero. From an experimental perspective, the conditions entering these
results depend only on the channel action on second moments, and the required matrices \(X\) and \(N_B\) can
therefore be inferred from covariance-matrix tomography.

Two other important zero-capacity classes are \emph{replacement channels}, which map every input to a fixed Gaussian state, and \emph{entanglement-breaking} (EB) channels, which destroy all input--output entanglement.  
A channel is EB if and only if its Choi state is separable~\cite{doi:10.1142/S0129055X03001709}.  
In the fermionic Gaussian setting, Gaussian separability is extremely restrictive: any ``measure-and-prepare'' map that outputs Gaussian states must prepare the \emph{same} Gaussian state for all measurement outcomes.  
Intuitively, this follows because fermionic Gaussian states have vanishing first moments (due to parity superselection) and mixtures of Gaussian states with distinct covariances are generically non-Gaussian.  
Consequently, the only Gaussian-preserving EB maps have $X=0$, so that $M_A \mapsto N_B$ is a replacement channel.  
We now give a proof using the extendibility framework.
\begin{proposition}
\label{proposition:EB_replacement}
For fermionic Gaussian channels, entanglement-breaking is equivalent to being a replacement channel.
\end{proposition}
\begin{proof}
    Applying Theorem~\ref{theorem-extendible} to the Choi CM $M_{\mathcal{N}}$ in Eq.~\eqref{fgc-cm} with $k_1=k$, $k_2=1$, $M_B=0$, and $M_A=N_B$, and using the Schur complement with respect to the bottom-right block gives $I + i N_B + i (k-1) Z - k X X^T \ge 0$. Taking $k\to\infty$ yields $X X^T \le i Z$; combined with $X X^T \le -i Z$, we conclude $X=0$. The channel then reduces to $M_A \mapsto N_B$, a replacement channel, which is trivially entanglement-breaking.
\end{proof}
The coincidence of EB and replacement channels may have interesting consequences in quantum information processing. In particular, it implies that, asymptotically, discriminating any (possibly non-Gaussian) fermionic channel from a fermionic Gaussian EB channel requires only non-adaptive tensor-power strategies~\cite{Cooney_2016}.  In particular, for the fermionic Gaussian attenuation channel with \(X=\sqrt{\lambda}\,I\) and \(N_B=0\), Proposition~\ref{proposition:EB_replacement} implies that the channel is entanglement breaking only at the trivial point \(\lambda=0\), where \(X=0\) and the map reduces to a replacement channel.

\emph{Discussion \& Outlook}: We have presented a comprehensive study of two-sided $(k_1,k_2)$-extendibility for fermionic Gaussian states, providing a full covariance-matrix characterization and a simple SDP formulation. Leveraging tools from quantum information theory, we derived finite de Finetti bounds that improve previous exponential scaling to linear in the number of modes, with complementary bounds for relative entropy and squashed entanglement. Through a Jordan–Wigner or related fermion-to-spin encoding, fermionic Gaussian states map to the free-fermion or matchgate sector of spin systems. Hence, after fixing the encoding convention, our covariance-matrix extendibility criterion can also be translated into an efficient criterion for this structured class of spin states, though not for arbitrary qubit states. These results provide operational criteria for separability in the fermionic Gaussian sector, ruling out nontrivial separable Gaussian states and identifying the definition that is both natural and practically meaningful. Finally, we established an SDP criterion for the anti-degradability of fermionic Gaussian channels.

Promising directions for future work include exploring applications of the derived bounds, investigating the family of classical–quantum binding states, and performing a systematic study of two-sided $(k_1,k_2)$-extendibility, which remains largely unexplored. Another further direction is to extend the analysis beyond strictly Gaussian states to convex-Gaussian states, i.e., convex mixtures of fermionic Gaussian states. Such states are generally not described by a single covariance matrix, but if the Gaussian ensemble is known, their multipoint correlators can still be computed efficiently as weighted sums of Wick/Pfaffian expressions. It would be interesting to generalize the present extendibility and de~Finetti results to this broader class. Finally, it is an open question whether fermionic Gaussian steerability can be rigorously defined and quantified, and how it is constrained by finite extendibility~\cite{Lami_2019,PhysRevLett.114.060403}.

\emph{Acknowledgements:}
We are grateful for discussions with Timothy Hsieh and Amin Moharramipour. This research was supported in part by the Perimeter Institute for Theoretical Physics, which is funded by the Government of Canada through the Department of Innovation, Science and Economic Development, and by the Province of Ontario through the Ministry of Colleges and Universities. A-R.N acknowledges support from the Natural Sciences and Engineering Research Council of Canada (NSERC) under Discovery Grant No. RGPIN-2018-04380, as well as from an Ontario Early Researcher Award. F.S.'s work is supported by the European Commission as a Marie Skłodowska-Curie Global Fellow, and he gratefully acknowledges the hospitality of the Institute for Quantum Computing (IQC), University of Waterloo, and Freie Universität Berlin.

\emph{Data availability.---}
The numerical data and source code supporting the disordered-Kitaev-chain results presented in the Supplemental Material are available in Ref.~\cite{NegariSalekData2026}.

\bibliographystyle{apsrev4-2}
\bibliography{references}

\clearpage
\newpage
\onecolumngrid

\setcounter{section}{0}
\setcounter{subsection}{0}
\setcounter{equation}{0}
\setcounter{figure}{0}
\setcounter{table}{0}
\setcounter{theorem}{0}
\renewcommand{\thesection}{S\arabic{section}}
\renewcommand{\theequation}{S\arabic{equation}}
\renewcommand{\thefigure}{S\arabic{figure}}
\renewcommand{\thetable}{S\arabic{table}}
\renewcommand{\thetheorem}{S\arabic{theorem}}

\begin{center}
{\large\bfseries Supplemental Material for\\[0.25em]
Extendibility of Fermionic Gaussian States\par}
\vspace{0.8em}
{\normalsize Amir-Reza Negari$^{1}$ and Farzin Salek$^{1,2,3}$\par}
\vspace{0.5em}
{\small
$^{1}$Perimeter Institute for Theoretical Physics, Ontario, Canada\\
$^{2}$Dahlem Center for Complex Quantum Systems, Freie Universit\"{a}t Berlin, Berlin, Germany\\
$^{3}$Institute for Quantum Computing (IQC), University of Waterloo, Ontario, Canada\par}
\end{center}
\vspace{1em}

\section{Preliminaries}

\subsection{Hilbert spaces, norms, and a convexity fact:}
The underlying space in our discussion is a Hilbert space, equipped with the norm induced by its inner product. For any vector $\ket{\psi}$,
\[
\|\psi\| = \sqrt{\braket{\psi|\psi}},
\]
and we call $\ket{\psi}$ normalized if $\|\psi\|=1$.

\medskip

For a linear operator $A$, the \emph{operator norm} is
\begin{align}
   \label{sm:operator-norm-def}
\|A\|_{\text{op}} = \sup_{\|\phi\|=1} \|A\ket{\phi}\| = \sup_{\|\phi\|=1} \sqrt{\bra{\phi}A^\dagger A\ket{\phi}}. 
\end{align}
When $A$ is Hermitian, the spectral theorem provides a decomposition $A = \sum_i \lambda_i \ket{v_i}\!\bra{v_i}$ with real eigenvalues $\lambda_i$. In this case
\[
\|A\|_{\text{op}} = \max_i |\lambda_i|.
\]

\medskip

More generally, any matrix $A$ (square or rectangular) admits a \emph{singular value decomposition (SVD)}
\[
A = U \Sigma V^\dagger,
\]
where $U,V$ are unitary (or isometries in the rectangular case) and $\Sigma$ is diagonal with nonnegative entries $s_i$ called the \emph{singular values} of $A$. By definition, $s_i$ are the eigenvalues of $|A| := \sqrt{A^\dagger A}$. This is distinct from the spectral decomposition, which applies only to normal operators. For Hermitian positive semidefinite $A$, the two coincide and the singular values are just the eigenvalues. In general, singular values equal the absolute values of eigenvalues only when $A$ is normal.

\medskip

The singular values provide a unified way to define matrix/operator norms, known as Schatten $p$-norms. For $1 \le p < \infty$,
\[
\|A\|_p := \big( \operatorname{Tr}[(A^\dagger A)^{p/2}] \big)^{1/p}
   = \left( \sum_i |s_i|^p \right)^{1/p},
\]
and for $p=\infty$,
\[
\|A\|_\infty := \lim_{p\to\infty} \|A\|_p = \max_i s_i.
\]
Special cases include:
\begin{itemize}
\item Operator norm: $\|A\|_{\text{op}} = \|A\|_\infty = \max_i s_i$,
\item Trace norm: $\|A\|_1 = \sum_i |s_i|$,
\item Hilbert–Schmidt norm: $\|A\|_2 = \sqrt{\sum_i s_i^2} = \sqrt{\operatorname{Tr}(A^\dagger A)}$.
\end{itemize}
These norms are always defined via singular values, which are nonnegative by construction.

\medskip

A standard fact from convex analysis is that any linear functional achieves its maximum over a compact convex set at one of the extreme points of that set. In quantum theory this principle is particularly useful: for instance, when optimizing a functional of the form 
\[
\sigma \;\mapsto\; \bra{\Psi}\sigma\ket{\Psi},
\]
where $\ket{\Psi}$ is fixed and $\sigma$ ranges over the convex set of separable states. Since the extreme points of the separable set are precisely the pure product states $\ket{\alpha}\otimes\ket{\beta}$, the optimization reduces to this smaller class. Entangled states lie outside the separable set and thus never contribute. Consequently, the maximum overlap of a bipartite pure state with a separable state is attained on a pure product state, and is known to equal the largest Schmidt coefficient of the given state.

\subsection{From qubits to fermions: modes, parity, and locality}\label{sm:subsec:qubits-to-fermions}

Fermionic systems are most naturally described in the \emph{second-quantized} (Fock-space) formalism, in close analogy with the bosonic framework in quantum optics. Two fermion-specific features shape locality and entanglement: (a) each mode is effectively a two-level system due to the Pauli exclusion principle, and (b) a \emph{parity superselection rule} forbids coherent superpositions of even- and odd-parity sectors. With these ingredients, fermionic modes can be treated as qubit-like degrees of freedom in Fock space, with anticommutation relations handled automatically by the formalism—for instance, $(\ket{00}+\ket{11})/\sqrt{2}$ plays the role of a fermionic Bell pair.

Entanglement in fermionic systems can be discussed from two complementary viewpoints: either between indistinguishable particles (first quantization) \cite{PhysRevA.64.022303,PhysRevA.72.022302}, or between distinguishable \emph{sets of modes} corresponding to spatial regions (second quantization) \cite{Ba_uls_2007,PhysRevA.97.042325,PhysRevA.87.022338}. The latter perspective is the natural one: fermions are indistinguishable and may delocalize freely, so entanglement is defined between mode sets rather than between identical particles. Assigning modes to observers (e.g., Alice’s modes in region $A$ and Bob’s in region $B$), correlations arise from observables on these subsystems. Unlike in bosonic or spin/qubit systems, however, observables supported on disjoint fermionic regions may anticommute, introducing algebraic constraints that make the analysis of locality and entanglement more subtle.

The translation from qubits to fermions can be consolidated into the following structural principles:

\paragraph*{(i) Mode dimensionality.} Each fermionic mode is two-dimensional, with basis $\{\ket{0},\ket{1}\}$ for empty/occupied. At this level, a single mode is isomorphic to a qubit.

\paragraph*{(ii) Parity superselection.} Physical states and operations preserve fermion-number parity. Coherent superpositions of different parity sectors are forbidden. For one mode, $\alpha\ket{0}+\beta\ket{1}$ is not physical, whereas $p\ketbra{0}{0}+(1-p)\ketbra{1}{1}$ is. For two modes, $(\ket{00}+\ket{11})/\sqrt{2}$ is allowed (even parity), while $(\ket{00}+\ket{01})/\sqrt{2}$ is not.

\paragraph*{(iii) Observable algebras and locality.} Creation/annihilation operators satisfy the canonical anticommutation relations
\begin{align}
\{\hat{c}_i,\hat{c}_j^\dagger\}=\delta_{ij},\qquad \{\hat{c}_i,\hat{c}_j\}=\{\hat{c}_i^\dagger,\hat{c}_j^\dagger\}=0,
\end{align}
and bare mode operators on disjoint regions do not commute. Physical measurements, however, are generated by \emph{parity-preserving} (even) operators, and these \emph{do} commute across spacelike-separated regions:
\begin{align}
[O_A,O_B]=0\quad\text{for } O_A,O_B \text{ even, supported on } A,B.
\end{align}
Entanglement and separability are therefore defined with respect to the even (parity-preserving) local algebras.

\paragraph*{(iv) Tensor products and operator ordering.} Fermionic operators act on the global antisymmetric Fock space. Writing $\hat{c}_1\otimes I$ or $I\otimes \hat{c}_2$ naively would enforce commutation rather than anticommutation. When a tensor-product representation is required (e.g., to map to spins), Jordan--Wigner or related encodings are employed to keep track of operator order and parity strings.

\paragraph*{(v) Subsystem parity.} Let $\hat N_A=\sum_{k=1}^{n_A} \hat{c}_k^\dagger \hat{c}_k$ be the number operator on a subsystem $A$ with $n_A$ modes. The fermion-number parity on $A$ is
\begin{align}\label{sm:eq:parity}
P^A = (-1)^{\hat N_A} = e^{i\pi \hat N_A} = \prod_{k=1}^{n_A}\!\bigl(1-2\,\hat{c}_k^\dagger \hat{c}_k\bigr),
\end{align}
which satisfies $(P^A)^2=I^A$. Physical states are block-diagonal in the eigenspaces of $P^A$. (Equivalently, using Majorana operators introduced in Sec.~\ref{sm:subsec:majorana-covariance}, one can write $P^A = i^{\,n_A}\prod_{k=1}^{n_A}\gamma_{2k-1}^A\gamma_{2k}^A$.)

These structural ingredients—mode dimensionality, parity superselection, algebraic locality, operator ordering, and subsystem parity—underlie the passage from qubits to fermions and govern the formulation of fermionic quantum information.

\subsection{Majorana operators and covariance matrices}\label{sm:subsec:majorana-covariance}

For $n$ fermionic modes with creation/annihilation operators $a_j^\dagger,a_j$, introduce $2n$ Hermitian Majorana operators
\begin{align}
\gamma_{2j-1}=\hat{c}_j+\hat{c}_j^\dagger,\qquad
\gamma_{2j}=i(\hat{c}_j-\hat{c}_j^\dagger),
\end{align}
which generate a real Clifford algebra:
\begin{align}
\{\gamma_p,\gamma_q\}=2\delta_{pq},\qquad \gamma_p^\dagger=\gamma_p.
\end{align}
Given a (mixed) state $\rho$, its (Majorana) CM $M\in\mathbb{R}^{2n\times 2n}$ is
\begin{align}
M_{pq}=\frac{i}{2}\,\Tr\!\big(\rho[\gamma_q,\gamma_p]\big)
=\begin{cases}
\Tr\!\left(i\gamma_q\gamma_p\,\rho\right), & q\neq p,\\[2pt]
0, & q=p.
\end{cases}
\end{align}
Thus $M$ is real and antisymmetric. Writing $iM$ (which is Hermitian), the \emph{bona fide} (physicality) condition for CMs is the spectral bound
\begin{align}\label{sm:eq:bf}
\operatorname{spec}(iM)\subseteq [-1,1],
\end{align}
equivalently, there exists $O\in SO(2n)$ and real numbers $\lambda_1,\ldots,\lambda_n\in[-1,1]$ such that
\begin{align}\label{sm:eq:can-form}
M=O\!\bigoplus_{j=1}^n\begin{bmatrix}0&\lambda_j\\ -\lambda_j&0\end{bmatrix}\!O^{T}.
\end{align}
In particular, the singular values of $M$ are at most $1$. For later use, note $(\gamma_p\gamma_q)^\dagger=-\gamma_p\gamma_q$ for $p\neq q$, so $(i\gamma_p\gamma_q)^2=I$ and its eigenvalues are $\pm 1$.
\subsection{Fermionic Gaussian states}\label{sm:subsec:ferm-gaussian}

A (parity-preserving) state $\rho$ on $n$ modes is called \emph{fermionic Gaussian} if it can be expressed as the (properly normalized) exponential of a quadratic form in Majoranas:
\begin{align}\label{sm:eq:gauss-def}
\rho \;=\; C\,\exp\!\left(\tfrac{i}{2}\,\gamma^{T} h\,\gamma\right),
\end{align}
with $h\in\mathbb{R}^{2n\times 2n}$ antisymmetric and $C$ determined by the normalization $\Tr\rho=1$.  
Equivalently, $\rho$ is Gaussian if and only if all higher-order correlators are fixed by the two-point ones via Wick’s theorem: for any even Majorana monomial $\gamma(x)$,
\begin{align}
\Tr\!\big(i^{|x|/2}\gamma(x)\,\rho\big)\;=\;\operatorname{Pf}\!\big(M[x]\big),
\end{align}
where $M[x]$ is the principal submatrix of the CM $M$ supported on the indices of $x$. Expectations of odd-weight monomials vanish.

\paragraph*{Purity and canonical form.}
For Gaussian states, the eigenvalues $\{\lambda_j\}$ in the canonical form \eqref{sm:eq:can-form} completely determine the spectrum of $\rho$. Purity is equivalent to
\begin{align}
M^2=-I \qquad \Leftrightarrow \qquad |\lambda_j|=1\ \ \forall j,
\end{align}
while mixed Gaussians have $|\lambda_j|<1$ for some $j$.  
In the basis that block-diagonalizes $M$, the state factorizes into a product over $n$ two-Majorana subsystems:
\begin{align}\label{sm:eq:block-prod}
\rho \;=\; \frac{1}{2^{n}}\prod_{j=1}^{n}\!\big(I+ i\lambda_j\,\tilde{\gamma}_{2j-1}\tilde{\gamma}_{2j}\big),
\end{align}
where $\tilde{\gamma}=O^{T}\gamma$.  
Using the identity
\[
\exp\!\big(\alpha\, i\tilde{\gamma}_{2j-1}\tilde{\gamma}_{2j}\big)
= \cosh\alpha \;+\; \sinh\alpha\, i\tilde{\gamma}_{2j-1}\tilde{\gamma}_{2j},
\]
and comparing with \eqref{sm:eq:gauss-def} gives the matrix functional relation
\begin{align}\label{sm:eq:tanh}
\tanh\!\left(\tfrac{h}{2}\right) = M
\end{align}
(in the simultaneous block-diagonal basis of $h$ and $M$).

\paragraph*{Single-mode states.}
Parity superselection fixes the structure of one-mode states: the only pure states are $\ket{0}$ and $\ket{1}$, and mixed states are diagonal in this basis. Every one-mode state is Gaussian and admits the representation
\begin{align}
\rho \;=\; \tfrac12\big(I+ i\lambda\,\gamma_1\gamma_2\big), \qquad \lambda\in[-1,1],
\end{align}
with $\lambda=\pm 1$ corresponding to $\ketbra{0}{0}$ and $\ketbra{1}{1}$, and $\lambda=0$ the maximally mixed state.

\paragraph*{Vacuum and Bell-pair examples.}
Let $\ket{\mathrm{vac}}$ denote the fermionic Fock vacuum. For a single mode with Majoranas $(\gamma_1,\gamma_2)$ its CM is
\begin{align}
M_{\ket{\mathrm{vac}}}=
\begin{bmatrix}
0 & -1\\
1 & \phantom{-}0
\end{bmatrix}.
\end{align}
For two modes with Majoranas $(\gamma_1,\gamma_2,\gamma_3,\gamma_4)$, the even Bell states
\begin{align}
\ket{\Phi^\pm}=\tfrac{1}{\sqrt{2}}(\ket{00}\pm \ket{11}),\quad
\ket{\Psi^\pm}=\tfrac{1}{\sqrt{2}}(\ket{01}\pm \ket{10})
\end{align}
have CMs (with respect to the natural ordering) given by
\begin{align}
M_{\Phi^\pm}=
\begin{bmatrix}
0 & 0 & 0 & \pm 1\\
0 & 0 & \pm 1 & 0\\
0 & \mp 1 & 0 & 0\\
\mp 1 & 0 & 0 & 0
\end{bmatrix},\qquad
M_{\Psi^\pm}=
\begin{bmatrix}
0 & 0 & 0 & \mp 1\\
0 & 0 & \pm 1 & 0\\
0 & \mp 1 & 0 & 0\\
\pm 1 & 0 & 0 & 0
\end{bmatrix}.
\end{align}
More generally, the maximally entangled $m$-mode state between two $m$-mode registers,
\begin{align}
\ket{\mathrm{EPR}} \;=\; \frac{1}{\sqrt{2^{m}}}\sum_{j=0}^{2^m-1}\ket{j,j},
\end{align}
has a CM with block off-diagonal form
\begin{align}
M_{\mathrm{EPR}}=
\begin{bmatrix}
0 & I_{m}\\
-\,I_{m} & 0
\end{bmatrix}.
\end{align}

\paragraph*{Locality, ordering, and signs.}
Because $\{\hat{c}_i^\dagger,\hat{c}_j^\dagger\}=0$, the order of creation operators is crucial:
\begin{align}
\hat{c}_1^\dagger \hat{c}_2^\dagger\ket{\mathrm{vac}} = -\,\hat{c}_2^\dagger \hat{c}_1^\dagger\ket{\mathrm{vac}}.
\end{align}
A consistent global convention for operator ordering must be fixed, after which relative minus signs in multi-mode states follow automatically. This is a bookkeeping feature of second quantization and does not obstruct operational locality, provided we restrict to even observables.

\paragraph*{Two structural corollaries.}
(i) Since $M$ is real and antisymmetric, its eigenvalues come in pairs $\pm i\lambda_j$ with $\lambda_j\in\mathbb{R}$. The physicality condition \eqref{sm:eq:bf} is equivalently $|\lambda_j|\le 1$.  
(ii) Traces of odd Majorana monomials vanish.
\medskip
\subsection{Fermionic Gaussian channels}

A \emph{fermionic Gaussian channel} (FGC) is a completely positive trace-preserving (CPTP) map that sends fermionic Gaussian states to fermionic Gaussian states~\cite{Greplov__2018,bravyi2005classicalcapacityfermionicproduct}. As in the bosonic case, FGCs are fully determined by their action on CMs.  

An $m\to n$ mode FGC $\mathcal{N}^{A\to B}$ is specified by a real $2n\times 2m$ matrix $X$ and a real antisymmetric $2n\times 2n$ matrix $N_B$, acting as
\begin{align}
\label{sm:G-channel}
    M_A \ \longmapsto\ X M_A X^T + N_B ,
\end{align}
for any input CM $M_A$. Bona fid constraints on $(X,N_B)$ can be most easily obtained via the Choi--Jamiołkowski (Choi) state, defined by 
\begin{align}
    C_{\mathcal{N}}^{AB} = (\mathcal{N}\otimes I_A)(\Phi_{AA'}^{+}),
\end{align}
with $\Phi_{AA'}^{+}$ a maximally entangled fermionic Gaussian Bell state. Since Gaussianity is preserved, $C_{\mathcal{N}}^{AB}$ is Gaussian with CM
\begin{align}
\label{sm:CJ-CM}
    M_{AB} =
    \begin{pmatrix}
        N_B & X \\
        -X^T & 0
    \end{pmatrix}.
\end{align}
The map \eqref{sm:G-channel} is a valid FGC iff $M_{AB}$ is a bona fide CM, i.e.
\begin{align}
\label{sm:iff-G-channel}
    I + iN_B - XX^T \ \ge\ 0 .
\end{align}
Implying: (i) $XX^T \le I$, and (ii) $I+iN_B \ge 0$.

\subsubsection*{Degradable, anti-degradable, and related channels}

An FGC $\mathcal{N}$ is \emph{degradable} if its complementary channel $\mathcal{N}^c$ can be simulated from $\mathcal{N}$ by a CPTP degrading map $\mathcal{D}$:
\begin{align}
    \mathcal{N}^c = \mathcal{D}\circ \mathcal{N}.
\end{align}
It is \emph{anti-degradable} if there exists a CPTP map $\mathcal{A}$ such that
\begin{align}
    \mathcal{N} = \mathcal{A}\circ \mathcal{N}^c,
\end{align}
in which case its quantum capacity vanishes by no-cloning theorem.  

A special case is the \emph{replacement channel}, which outputs a fixed state $\sigma$ regardless of input: 
\begin{align}
    \mathcal{N}(\rho) = \sigma, \quad \forall\rho.
\end{align}
If $\sigma$ is separable from any reference, $\mathcal{N}$ is \emph{entanglement-breaking}~\cite{doi:10.1142/S0129055X03001709}, destroying all input–reference entanglement.

\subsection{useful lemmas:}

\begin{lemma}[{\cite[Proposition~3]{trace-antonio}}]
\label{sm:trace-distance-lemma}
Let \( \rho \) and \( \sigma \) be two fermionic states with correlation matrices \( M_\rho \) and \( M_\sigma \), respectively. Then, their trace distance is lower bounded by the operator norm of the difference of their correlation matrices:
\begin{equation}
\| M_\rho - M_\sigma \|_\infty \leq \| \rho - \sigma \|_1.
\end{equation}
\end{lemma}
\medskip

\begin{lemma}[{\cite[Theorem 6]{trace-antonio}}]
\label{sm:trace-distance-lemma-upper}
Let $\rho$ and $\sigma$ be two fermionic states (possibly mixed) with CMs $M_\rho$ and $M_\sigma$, respectively.  
Then, the trace distance between $\rho$ and $\sigma$ is upper bounded by one half of the trace norm of the difference of their CMs:
\begin{equation}
\| \rho - \sigma \|_1 \leq \frac{1}{2} \| M_\rho - M_\sigma \|_1.
\end{equation}
\end{lemma}

\begin{lemma}
\label{sm:column-sum}
    Given a CM \(M\) the following inequality holds:
    \begin{equation*}
        \sum_k (M_{ki})^2 \leq 1,
    \end{equation*}
\end{lemma}
\begin{proof}
Since \(M\) is a valid fermionic covariance matrix, all singular values of \(M\) are at most one. Hence
\[
        M^TM\le I .
\]
Taking the \(i\)-th diagonal entry of this matrix inequality gives
\[
        (M^TM)_{ii}=\sum_k M_{ki}^2\le 1,
\]
which proves the claim.
\end{proof}

\begin{lemma}
\label{sm:abs-val-1}
Let \( M \) be the CM of a fermionic Gaussian state. Then, for all \( i, j \),
\begin{equation*}
    |M_{ij}| \leq 1.
\end{equation*}
\end{lemma}

\begin{proof}
By definition, the off-diagonal entries of the CM \( M \) are given by
\[
M_{ij} = \operatorname{Tr}(i \gamma_j \gamma_i \rho),
\]
where \( \gamma_i \) and \( \gamma_j \) are Majorana operators satisfying \( \gamma_k = \gamma_k^\dagger \), \( \gamma_k^2 = I \), and \( \{ \gamma_i, \gamma_j \} = 2\delta_{ij} \).

Define the operator \( A := i \gamma_j \gamma_i \). Then \( A \) is Hermitian, and since
\[
A^2 = (i \gamma_j \gamma_i)^2 = -\gamma_j \gamma_i \gamma_j \gamma_i = -\gamma_j (\gamma_i \gamma_j) \gamma_i = -\gamma_j (-\gamma_j \gamma_i) \gamma_i = \gamma_j \gamma_j \gamma_i \gamma_i = I,
\]
we have \( A^2 = I \), so the eigenvalues of \( A \) are \( \pm 1 \). Therefore, its operator norm is \( \|A\|_{\text{op}} = 1 \).

Now, since \( \rho = \sum_kp_k\ketbra{\psi_k} \) is a valid density matrix (positive semidefinite, \( \operatorname{Tr}(\rho) = 1 \)), the expectation value satisfies
\begin{align*}
    |M_{ij}| &= |\Tr(A \rho)| \\
&= \abs{\sum_kp_k\bra{\psi_k}A\ket{\psi_k}}\\
&\leq \sum_kp_k\abs{\bra{\psi_k}A\ket{\psi_k}} \\
&\leq  \sum_kp_k . \|A\|_{\text{op}} \\
&= \|A\|_{\text{op}} \\
&= 1,
\end{align*}
where the first inequality follows from Jensen's inequality applied to the absolute value function which is a convex function (it also follows more easily from triangle inequality though), and the second inequality follows from the definition of the operator norm Eq. \eqref{sm:operator-norm-def}.
This proves the claim.
\end{proof}

\medskip

\begin{lemma}
\label{sm:convex-set}
For any bipartite pure state $\ket{\psi}^{AB}$ with maximal Schmidt coefficient $\lambda_1(\psi)$, the maximal overlap with a separable state $\sigma^{AB}$ is bounded by
\[
\bra{\psi}\,\sigma\,\ket{\psi} \;\leq\; \lambda_1(\psi).
\]
\end{lemma}

\begin{proof}
The functional $\sigma \mapsto \bra{\psi}\sigma\ket{\psi}$ is linear in $\sigma$. By convex analysis, its maximum over the compact convex set of separable states 
\[
\sigma^{AB} = \sum_x p(x)\,\sigma_x^A \otimes \sigma_x^B
\]
is attained at an extreme point, i.e., at a pure product state $\ket{\alpha}\otimes\ket{\beta}$. The maximum squared overlap of $\ket{\psi}$ with a product state is known to equal the largest Schmidt coefficient $\lambda_1(\psi)$. Therefore,
\[
\max_{\text{sep }\sigma} \bra{\psi}\sigma\ket{\psi} \;\leq\; \lambda_1(\psi).
\]
\end{proof}

\medskip
\section{Extendibility of Gaussian fermions}
\label{sm:extendibility-of-fermions}

Let \(\rho^{AB}\) be a fermionic Gaussian state with CM
\begin{equation*}
    M_{AB} =
    \begin{pmatrix}
        M_A & X \\
        -X^{T} & M_B
    \end{pmatrix}.
\end{equation*}
By definition, \(M_{AB}\) is real and antisymmetric. The diagonal blocks \(M_A\) and \(M_B\) are the reduced CMs of the subsystems \(A\) and \(B\), hence are real and antisymmetric themselves. The off-diagonal block is otherwise unconstrained apart from the requirement that the full CM remain real and antisymmetric; accordingly, we may take \(X\) to be an arbitrary real matrix (in particular, \(X\) need not be antisymmetric).

The set of Gaussian states is closed under taking extensions by Theorem 1 of the main text; in particular, a \((k_1,k_2)\)-extension of \(\rho^{AB}\) exists within the Gaussian family and has a CM of the form given in Eq.~\eqref{sm:eq:ext_CM}. Whenever that extended CM is valid, all moments (including higher-order ones) are finite; see, e.g., Lemma~\ref{sm:abs-val-1}. The crucial observation, however, is that the additional blocks \(Y\) and \(Z\) must themselves be antisymmetric, a property that is far from obvious a~priori. Recognizing and proving this fact is essential in the fermionic setting. In the bosonic case, the corresponding statement was established in~\cite{Lami_2019} by combining extendibility with structural properties of the bosonic CM. Our approach takes a different route: the result is obtained directly from the fermionic algebra of Majorana operators together with the definition of extendibility. This highlights that the fermionic case cannot simply be reduced to the bosonic one, but requires its own algebraic foundations.

\begin{lemma}
\label{sm:Y-antisymmetric}
The matrices \(Y\) and \(Z\) in Eq.~\eqref{sm:eq:ext_CM} is antisymmetric.
\end{lemma}

\begin{proof}
Consider two distinct copies \(B_m\) and \(B_n\) within the \(B\)-part of a symmetric \((k_1,k_2)\)-extension of \(\rho^{AB}\). Let \(\gamma_q^{B_m}\) and \(\gamma_p^{B_n}\) denote Majorana operators acting on \(B_m\) and \(B_n\), respectively. Writing (by slight abuse of notation) \(M_{(qp)} \equiv Y_{qp}\) for the \((q,p)\)-entry of the CM block coupling \(B_m\) to \(B_n\), we have
\begin{align*}
Y_{qp}
&= i\,\Tr\!\left(\gamma_q^{B_m}\gamma_p^{B_n}\,\rho^{A_1\ldots A_{k_1},B_1\ldots B_{k_2}}\right) \\
&= \Tr\!\left(i\,\gamma_q^{B_m}\gamma_p^{B_n}\,U_{mn}\,\widetilde{\rho}^{A_1\ldots A_{k_1},B_1\ldots B_{k_2}}\,U_{mn}^\dagger\right) \\
&= \Tr\!\left(i\,U_{mn}^\dagger\gamma_q^{B_m}U_{mn}\;U_{mn}^\dagger\gamma_p^{B_n}U_{mn}\;\widetilde{\rho}^{A_1\ldots A_{k_1},B_1\ldots B_{k_2}}\right) \\
&= \Tr\!\left(i\,\gamma_q^{B_n}\gamma_p^{B_m}\,\widetilde{\rho}^{A_1\ldots A_{k_1},B_1\ldots B_{k_2}}\right) \\
&= -\,\Tr\!\left(i\,\gamma_p^{B_m}\gamma_q^{B_n}\,\widetilde{\rho}^{A_1\ldots A_{k_1},B_1\ldots B_{k_2}}\right) \\
&= -\,Y_{pq}.
\end{align*}
Here \(U_{mn}\) is the unitary that swaps the subsystems \(B_m\) and \(B_n\) (an element of the permutation representation of \(S_{k_2}\)). The second line uses symmetric extendibility; the third uses cyclicity of the trace; the fourth uses \(U_{mn}^\dagger\gamma_q^{B_m}U_{mn}=\gamma_q^{B_n}\) and \(U_{mn}^\dagger\gamma_p^{B_n}U_{mn}=\gamma_p^{B_m}\); and the penultimate equality follows from the anticommutation of Majorana operators acting on disjoint subsystems. Hence \(Y_{qp}=-Y_{pq}\), i.e., \(Y\) is antisymmetric. The same reasoning applies to \(Z\), which is therefore antisymmetric as well.

\end{proof}

\medskip

The permutation-invariance requirement in the definition of symmetric extendibility can in fact be eliminated. One may define a \((k_1,k_2)\)-extension of \(\rho^{AB}\) simply as a state \(\widetilde{\rho}^{A_1\ldots A_{k_1},B_1\ldots B_{k_2}}\) whose two-body marginals satisfy \(\rho^{A_iB_j}=\rho^{AB}\) for all pairs \((i,j)\), without any symmetry assumption. With this weaker definition, extendibility and symmetric extendibility turn out to be equivalent. The implication from symmetric to general is trivial, while the converse follows by averaging: given any extension \(\widetilde{\rho}^{A_1\ldots A_{k_1},B_1\ldots B_{k_2}}\), one obtains a symmetric extension by twirling it over the product of permutation groups \(S_{k_1}\times S_{k_2}\),
\begin{align*}
\widetilde{\rho}^{A_1\ldots A_{k_1},B_1\ldots B_{k_2}}
=\frac{1}{k_1!\,k_2!}\sum_{\pi_a\in S_{k_1}}\sum_{\pi_b\in S_{k_2}}
\Big(U^{A_1\ldots A_{k_1}}_{\pi_a}\otimes U^{B_1\ldots B_{k_2}}_{\pi_b}\Big)\,
\widetilde{\widetilde{\rho}}^{A_1\ldots A_{k_1},B_1\ldots B_{k_2}}\,
\Big(U^{A_1\ldots A_{k_1}}_{\pi_a}\otimes U^{B_1\ldots B_{k_2}}_{\pi_b}\Big)^\dagger,
\end{align*}
where \(U^{A_1\ldots A_{k_1}}_{\pi_a}\) and \(U^{B_1\ldots B_{k_2}}_{\pi_b}\) implement the permutations \(\pi_a\) and \(\pi_b\).  

As a consequence, for a Gaussian bipartite state \(\rho^{AB}_{\mathrm{G}}\) it makes no difference whether one demands Gaussianity, symmetry, or both: the existence of a Gaussian symmetric \((k_1,k_2)\)-extension is equivalent to the existence of a Gaussian extension without symmetry, and both are in turn equivalent to the existence of a symmetric or general \((k_1,k_2)\)-extension that need not be Gaussian.

\section{Proof of Theorem 2}
\label{sm:appendix:two_sided_proof}

\medskip

Let $\rho^{A_1 \dots A_{k_1} B_1 \dots B_{k_2}}$ be a symmetric $(k_1,k_2)$-extension of the bipartite Gaussian state $\rho^{AB}_{\mathrm{G}}$.
Even if $\widetilde{\rho}$ is not Gaussian, its second moments define a CM of the form
\begin{align}
\nonumber
    M_{A_1 \dots A_{k_1} B_1 \dots B_{k_2}}
    &=
    \begin{pmatrix}
      M_{AA} & \widetilde{M}_{AB} \\ \label{sm:eq:ext_CM}
      -M_{AB}^{T} & M_{BB}
    \end{pmatrix}_{2\times 2}
    \\
    &=
    \begin{pmatrix}
      M_A     & Z       & \cdots & Z       & X       & X       & \cdots & X \\
      Z       & M_A     & \cdots & \vdots  & X       & X       & \cdots & X \\
      \vdots  & \vdots  & \ddots & Z       & \vdots  & \vdots  & \ddots & \vdots \\
      Z       & \cdots  & Z      & M_A     & X       & X       & \cdots & X \\[6pt]
      - X^T   & - X^T   & \cdots & - X^T   & M_B     & Y       & \cdots & Y \\
      - X^T   & - X^T   & \cdots & - X^T   & Y       & M_B     & \cdots & Y \\
      \vdots  & \vdots  & \ddots & \vdots  & \vdots  & \vdots  & \ddots & \vdots \\
      - X^T   & - X^T   & \cdots & - X^T   & Y       & Y       & \cdots & M_B
    \end{pmatrix}_{(k_1+k_2)\times (k_1+k_2)},
\end{align}
where permutation symmetry of the extension enforces the block-constant structure
\begin{align}
\label{sm:eq:matrices_compact}
M_{AA} =
\begin{pmatrix}
    M_A & Z & \dots & Z \\
    Z & M_A & \dots & \vdots \\
    \vdots & \vdots & \ddots & Z \\
    Z & \dots & Z & M_A
\end{pmatrix}_{k_1\times k_1},
\quad
M_{BB} =
\begin{pmatrix}
    M_B & Y & \dots & Y \\
    Y & M_B & \dots & \vdots \\
    \vdots & \vdots & \ddots & Y \\
    Y & \dots & Y & M_B
\end{pmatrix}_{k_2\times k_2},
\end{align}
\begin{align}
\label{sm:eq:ABblock}
M_{AB} =
\begin{pmatrix}
    X & X & \dots & X \\
    X & X & \dots & \vdots \\
    \vdots & \vdots & \ddots & X \\
    X & \dots & X & X
\end{pmatrix}_{k_1\times k_2}.
\end{align}
Here $X\in\mathbb{R}^{2n_A\times 2n_B}$ is arbitrary, while 
$Y\in\mathbb{R}^{2n_B\times 2n_B}$ and 
$Z\in\mathbb{R}^{2n_A\times 2n_A}$ are required to be real antisymmetric matrices.
 physically, $X,Y,Z$ encode inter-mode correlations across distinct fermionic subsystems.

An extension with CM \( M\) is physical iff the fermionic physicality condition holds:
\begin{align}
\label{sm:eq:phys_block}
I + i\widetilde{M}
=
\begin{pmatrix}
I + iM_{AA} & iM_{AB} \\
- iM_{AB}^{T} & I + iM_{BB}
\end{pmatrix} \ge 0.
\end{align}

To analyse \eqref{sm:eq:phys_block} we use the standard Schur-complement characterization of positivity of Hermitian block matrices \cite[Theorem~1.12]{Horn2005}:
\begin{align*}
  \begin{pmatrix} A & C \\ C^\dagger & B \end{pmatrix} \ge 0 
  \quad \Leftrightarrow \quad  
  A \ge 0,\quad B \ge C^\dagger A^{-1} C ,
\end{align*}
with the non-invertible case handled by the usual limiting argument \cite{Lami_2019}. The same equivalence holds with $A$ and $B$ interchanged. It is useful to write Eq. \eqref{sm:eq:phys_block} as follows:
\begin{align*}
I + iM
=
\begin{pmatrix}
I + iM_{AA} & iM_{AB} \\
(iM_{AB})^{\dagger} & I + iM_{BB}
\end{pmatrix} \ge 0.
\end{align*}

The Schur complement then reduces to the following condition:
\begin{align}
\label{sm:eq:schur-1}
I + iM_{AA} &\ge 0,\\[4pt]
\label{sm:eq:schur-2}
I + iM_{BB} - M_{AB}^\dagger (I + iM_{AA})^{-1} M_{AB} &\ge 0.
\end{align}
(These in particular imply $I + iM_{BB}\ge 0$ as remarked earlier.)

It is convenient to decompose the block-constant matrices into the symmetric direction (span of the uniform vector) and its orthogonal complement. Define
\[
\ket{+}_{k} \coloneqq \frac{1}{\sqrt{k}}\sum_{j=1}^{k}\ket{j}
\]
A direct expansion yields the useful identities (we keep the explicit intermediate step because it is instructive):
\begin{align}
\nonumber
I + iM_{AA}
&=
\begin{pmatrix}
I+iM_A & iZ & \cdots & iZ \\
iZ & I+iM_A & \cdots & \vdots \\
\vdots & \vdots & \ddots & iZ \\
iZ & \cdots & iZ & I+iM_A
\end{pmatrix}_{2k_1\times 2k_1} \\[4pt]
\label{sm:first-schur-used}
&= k_1\ketbra{+}_{k_1}\otimes iZ + I_{k_1} \otimes (I+iM_A - iZ) \\[4pt]
\label{sm:first-schur}
&= \ketbra{+}_{k_1}\otimes\bigl(I+iM_A +(k_1-1) iZ\bigr) + \bigl(I_{k_1} - \ketbra{+}_{k_1}\bigr)\otimes\bigl(I+iM_A - iZ\bigr).
\end{align}
The analogous decomposition on the $B$ side is
\begin{align}
\label{sm:second-schur}
I + iM_{BB}
= \ketbra{+}_{k_2}\otimes\bigl(I + iM_B + (k_2 - 1)iY\bigr)
+ \bigl(I_{k_2} - \ketbra{+}_{k_2}\bigr)\otimes\bigl(I + iM_B - iY\bigr).
\end{align}
Finally the off-diagonal block is proportional to the cross projector:
\begin{align}
\label{sm:third-schur}
iM_{AB} = \sqrt{k_1 k_2}\; \ket{+}_{k_1}\!\bra{+}_{k_2} \otimes iX.
\end{align}

Because $\ketbra{+}$ and $I-\ketbra{+}$ are orthogonal projectors, the large matrices above decompose into direct sums acting on orthogonal subspaces. Consequently the Schur conditions \eqref{sm:eq:schur-1}--\eqref{sm:eq:schur-2} reduce to the following four matrix inequalities, obtained by imposing positivity separately on the symmetric sector and on its orthogonal complement:

\begin{align}
\label{sm:schur-c1}
I + iM_A +(k_1-1) iZ &\ge 0,\\
\label{sm:schur-c2}
I + iM_A - iZ &\ge 0,\\
\label{sm:eq:schur_final}
I + iM_B + (k_2 - 1)iY - k_1 k_2\, X^T \bigl(I + iM_A + (k_1 - 1)iZ\bigr)^{-1} X &\ge 0,\\
\label{sm:schur-c4}
I + iM_B - iY &\ge 0.
\end{align}

Thus \eqref{sm:schur-c1}--\eqref{sm:schur-c4} are necessary conditions for the CM of a $(k_1,k_2)$-extension to be physical. Moreover, when the global state is Gaussian these conditions are also sufficient (second moments completely characterize Gaussian states).

These condition can be succinctly imposed by asking the following matrix to satisfy the fermionic bona fide condition:

\begin{align}
    \label{sm:eq:new_CM}
 \begin{pmatrix} 
        (M_A-Z) & 0 & 0 & 0 \\
        0 & M_A + (k_1 - 1)Z & \sqrt{k_1 k_2} X & 0 \\
        0 & -\sqrt{k_1 k_2} X^{T} & M_B + (k_2 - 1)Y & 0 \\
        0 & 0 & 0 & (M_B-Y)
    \end{pmatrix} 
\end{align}

Define $\Delta_A = M_{A} - Z$ and $\Delta_B = M_B - Y$, and we rewrite the matrix as follows:

\begin{align}
 \begin{pmatrix} 
        \Delta_A & 0 & 0 & 0 \\
        0 & k_1M_A - (k_1 - 1)\Delta_A & \sqrt{k_1 k_2} X & 0 \\
        0 & -\sqrt{k_1 k_2} X^{T} & k_2M_B - (k_2 - 1)\Delta_B & 0 \\
        0 & 0 & 0 & \Delta_B
    \end{pmatrix} 
\end{align}
Hence, a necessary condition for $(k_1,k_2)$-extendibility of any fermionic state is existence of valid CMs $\Delta_A$ and $\Delta_B$ such that the following is a valid CM:
\begin{align}
    \label{sm:eq:new_CM-new-new}
 \begin{pmatrix} 
        k_1M_A - (k_1 - 1)\Delta_A & \sqrt{k_1 k_2} X  \\
     -\sqrt{k_1 k_2} X^{T} & k_2M_B - (k_2 - 1)\Delta_B
    \end{pmatrix} 
\end{align}
When $\rho^{AB}$ is Gaussian, second moments are sufficient to characterize extendibility, and the condition above becomes necessary and sufficient.
\hfill$\square$

\section{Proof of Lemma~3 of the main text}

\emph{Proof.—}
Assume that \(\rho^{AB}\) is \((k_1,k_2)\)-extendible. By Theorem~2 of the main text, there exist valid CMs \(\Delta_A\) and \(\Delta_B\) such that
\begin{align}
M_{\mathrm{eff}}
=
\begin{pmatrix}
k_1M_A-(k_1-1)\Delta_A & \sqrt{k_1k_2}\,X\\
-\sqrt{k_1k_2}\,X^T & k_2M_B-(k_2-1)\Delta_B
\end{pmatrix}
\end{align}
is a valid fermionic CM. Hence all singular values of \(M_{\mathrm{eff}}\) are at most one, or equivalently
\begin{align}
        M_{\mathrm{eff}}^T M_{\mathrm{eff}}\le I .
\end{align}

Taking the principal block of this inequality corresponding to the \(B\)-subsystem gives
\begin{align}
        k_1k_2\,X^T X
        +
        \bigl(k_2M_B-(k_2-1)\Delta_B\bigr)^T
        \bigl(k_2M_B-(k_2-1)\Delta_B\bigr)
        \le I .
\end{align}
Since the second term is positive semidefinite, we obtain
\begin{align}
        X^T X\le \frac{I}{k_1k_2}.
\end{align}
Equivalently,
\begin{align}
        \|X\|_{\mathrm{op}}\le \frac{1}{\sqrt{k_1k_2}} .
\end{align}
Taking instead the principal block corresponding to the \(A\)-subsystem gives
\begin{align}
        k_1k_2\,XX^T
        +
        \bigl(k_1M_A-(k_1-1)\Delta_A\bigr)
        \bigl(k_1M_A-(k_1-1)\Delta_A\bigr)^T
        \le I ,
\end{align}
and therefore
\begin{align}
        XX^T\le \frac{I}{k_1k_2}.
\end{align}
This proves the claimed singular-value bound.
\hfill\(\square\)

\section{Separability in Fermionic Gaussian States and Theorem~5}
\label{sm:sec:fermion-separability}

In fermionic systems, the notion of a product state is more subtle than in qubit systems due to anti-commutation relations between operators on different subsystems \cite{Ba_uls_2007}. We briefly distill the relevant notions from \cite{Ba_uls_2007}. As an example bipartite fermionic state $\rho^{AB}$ is a \emph{product state} if it can be written as $\rho^{AB}=\rho^{A}\otimes \rho^{B}$. Following \cite{Ba_uls_2007}, we denote the set of such states by $\mathcal{P}2$. Here, neither the marginals $\rho^{A},\rho^{B}$ nor the global state $\rho^{AB}$ are required to be \emph{physical}, i.e., they may fail to commute with the parity operator. If we further require these states to be physical, the resulting set is denoted by $\mathcal{P}2_\pi$.
A different definition arises from factorization of expectation values: a fermionic state $\rho^{AB}$ is a product if, for all (possibly non-physical and non-commuting) observables $O^{A}$ and $O^{B}$,
\begin{align}
    \Tr\!\big[\rho^{AB}(O^{A} O^{B})\big] \;=\; \Tr(\rho^{A} O^{A})\,\Tr(\rho^{B} O^{B}).
\end{align}
Let $\mathcal{P}3$ denote the set of states satisfying this factorization. In general, $\mathcal{P}2 \neq \mathcal{P}3$. If, however, we restrict the \emph{states} (not the observables) to be physical, we obtain $\mathcal{P}3_\pi$, and the crucial identity holds:
\begin{align}
    \mathcal{P}2_\pi \;=\; \mathcal{P}3_\pi .
\end{align}
The proof is given in \cite[App.~A.6]{Ba_uls_2007} (see also \cite{moriya2005}). For intuition, we sketch an analogous argument in the qubit setting.

\medskip

\noindent\textbf{Qubit analogy.}
In qubit systems, a bipartite state $\rho^{AB}$ is product iff
\begin{align}
\label{sm:product-qubit-1}
 \Tr\!\big[\rho^{AB}(O_A \otimes O_B)\big] \;=\; \Tr(\rho^A O_A)\,\Tr(\rho^B O_B) \quad \forall\, O_A,O_B.
\end{align}
One direction is immediate. For the converse, recall that $\rho=\sigma$ iff $\Tr(\rho O)=\Tr(\sigma O)$ for all $O$. Since any $O^{AB}$ is a linear combination of product operators $O_A\otimes O_B$, Eq.~\eqref{sm:product-qubit-1} is equivalent to
\begin{align}
\label{sm:product-qubit-2}
 \Tr\!\big[\rho^{AB} O^{AB}\big] \;=\; \Tr(\rho^A O^A)\,\Tr(\rho^B O^B) \quad \forall\, O^{AB},
\end{align}
and \eqref{sm:product-qubit-2} implies \eqref{sm:product-qubit-1} trivially. This is the essence of the equivalence $\mathcal{P}2_\pi=\mathcal{P}3_\pi$.

\medskip

\noindent\textbf{Separable sets.}
There are two natural separability constructions. One may take convex combinations of $\mathcal{P}2$ states and require the mixture to be physical; denote this by $\mathcal{S}2'_\pi$. Alternatively, one may take convex combinations of $\mathcal{P}2_\pi(=\mathcal{P}3_\pi)$; denote this by $\mathcal{S}2_\pi$. In $\mathcal{S}2'_\pi$, the components may be unphysical but the mixture is physical, whereas in $\mathcal{S}2_\pi$ each component is already physical. Clearly, $\mathcal{S}2'_\pi \supset \mathcal{S}2_\pi$.

For \emph{Gaussian} fermionic states, separability relative to the smaller set $\mathcal{S}2_\pi$ exhibits a striking feature: Gaussian separable states are \emph{quantum-classical binding}, in the sense that classical correlations cannot appear without quantum correlations. In particular, one cannot generate classical correlations by Gaussian LOCC. One proof uses Theorem~5; another follows directly from
\begin{align}
    G\mathcal{S}3_\pi \;=\; G\mathcal{P}2_\pi \;=\; G\mathcal{P}3_\pi,
\end{align} 
which precludes classical correlations in the absence of entanglement under Gaussian operations and classical communication.

\noindent We now present the proof of Theorem~5, establishing that, for fermionic Gaussian states, classical correlations are present if and only if entanglement is present. This theorem, together with the extendibility framework, indicates that $\mathcal{S}2_{\pi}$ is the natural candidate for defining separable fermionic Gaussian states. Furthermore, it lends support to a more general picture: even beyond the Gaussian setting, extendibility requires the second-order moments to vanish, consistent with the intuition behind this definition.

\begin{proof}[Proof of Theorem~5]
We combine the fermionic superselection rule with a parity-twirl argument to show that all two-point cross-correlations between regions $A$ and $B$ vanish for separable, physical states; by Wick’s theorem, this then eliminates \emph{all} correlations in the Gaussian case.

Consider a separable fermionic state with decomposition
\begin{equation*}
\rho^{AB} = \sum_i p_i\, \rho^A_i \otimes \rho^B_i .
\end{equation*}
A sufficient condition for $\rho^{AB}$ to satisfy the fermionic superselection rule is that each component respects parity:
\begin{equation*}
[\rho_i^{AB},P^{AB}] = 0 \quad \text{with} \quad \rho_i^{AB}:=\rho^A_i\otimes\rho^B_i .
\end{equation*}
Equivalently, each factor commutes with the local parity,
\begin{equation*}
[\rho^A_i,P^A]=0,\qquad [\rho^B_i,P^B]=0 \quad \forall i.
\end{equation*}

Let $\gamma_m^A$ and $\gamma_n^B$ be Majorana operators on $A$ and $B$. Using $P^A \gamma_m^A P^A = -\gamma_m^A$, $(P^A)^2 = I^A$, and the cyclicity of trace, we compute
\begin{align*}
\Tr\!\big[(\gamma_m^A \gamma_n^B)\, \rho^{AB}\big]
&= \sum_i p_i\, \Tr\!\big[(\gamma_m^A  \gamma_n^B)(\rho^A_i \otimes \rho^B_i)\big] \\
&= \sum_i p_i\, \Tr\!\big[P^A(\gamma_m^A  \gamma_n^B)(\rho^A_i \otimes \rho^B_i)P^A\big] \\
&= \sum_i p_i\, \Tr\!\big[-(\gamma_m^A  \gamma_n^B) \, P^A(\rho^A_i \otimes \rho^B_i)P^A\big] \\
&= - \sum_i p_i\, \Tr\!\big[(\gamma_m^A  \gamma_n^B)(\rho^A_i \otimes \rho^B_i)\big] \\
&= - \Tr\!\big[(\gamma_m^A \gamma_n^B)\, \rho^{AB}\big],
\end{align*}
where in the third line we used $P^A \gamma_m^A P^A = -\gamma_m^A$, and in the fourth line we used that $P^A$ commutes with $\rho^A_i \otimes \rho^B_i$ because $[\rho^A_i,P^A]=0$. Hence
\begin{equation*}
\Tr\!\big[(\gamma_m^A \gamma_n^B)\, \rho^{AB}\big] = 0 .
\end{equation*}

Therefore, in any separable fermionic state obeying the superselection rule, all two-point cross-correlations between $A$ and $B$ vanish. For a fermionic \emph{Gaussian} state $\rho_{\text{G}}^{AB}$, Wick’s theorem implies that second moments determine all higher moments; thus the absence of cross-correlations implies the absence of correlations of any order between $A$ and $B$, and consequently
\begin{equation}
    \rho_{\text{G}}^{AB} = \rho^A_{\text{G}} \otimes \rho^B_{\text{G}} .
\end{equation}
This conclusion is specific to Gaussian states: in general, non-Gaussian separable states may exhibit classical correlations beyond second order even in the absence of entanglement.
\end{proof}

\begin{corollary}
A fermionic Gaussian state with CM $M_{AB}$ is separable if and only if $M_{AB} = M_A \oplus M_B$.
\end{corollary}

\begin{proof}
The forward direction is immediate. For the converse, Theorem~5 implies that all two-point correlations between $A$ and $B$ vanish for separable Gaussian states, so the CM is block-diagonal, $M_{AB}=M_A \oplus M_B$.
\end{proof}

\begin{remark}
This motivates a particularly simple entanglement indicator for fermionic Gaussian states:
\begin{equation}
E_{\mathrm{cq}}(\rho_{\text{G}}^{AB}) := \tfrac{1}{2}\, I(A;B)_{\rho_{\text{G}}},
\end{equation}
where $I(A;B)=S(A)+S(B)-S(AB)$ is the quantum mutual information. This quantity
(i) vanishes iff the state is separable,
(ii) coincides with the entropy of entanglement for pure states,
(iii) is efficiently computable from the CM, and
(iv) is monotone under Gaussian LOCC.
In this sense, $E_{\mathrm{cq}}$ cleanly captures the \emph{classical--quantum dichotomy} of fermionic Gaussian correlations: every state is either a product state or contains genuine quantum entanglement. We investigate this as an entanglement measure for fermionic Gaussian states in a forthcoming paper.
\end{remark}

\medskip

\section{Proof of Theorem~6}

Before presenting the proof, we contrast our setting with the related bosonic case~\cite[Theorem~3]{Lami_2019}, noting that while their framework relies on the CM structure specific to bosons, our analysis proceeds differently in the fermionic setting. In the bosonic context, restricting to Gaussian separable states, one can always find a separable Gaussian state whose CM differs from that of the given state by a constant factor. Owing to the commutativity of the two CMs, related techniques can then be applied. The situation for fermionic Gaussian states is markedly different, as we now demonstrate.

\subsection*{Proof of (i): Trace distance}
Let \(\rho^A = \operatorname{tr}_B(\rho^{AB})\) and \(\rho^B = \operatorname{tr}_A(\rho^{AB})\). Since \(\rho^A \otimes \rho^B \in \mathrm{SEP}(A\!:\!B)\), we immediately obtain
\begin{equation}
    \| \rho^{AB}_{G} - \mathrm{SEP}(A\!:\!B) \|_1 
    \;\le\; \| \rho^{AB}_{G} - \rho^A_{\text{G}} \otimes \rho^B_{\text{G}} \|_1 .
\end{equation}
The remaining task is to bound 
\(\| \rho^{AB}_{G} - \rho^A_{\text{G}} \otimes \rho^B_{\text{G}} \|_1\).  
By Lemma~\ref{sm:trace-distance-lemma-upper}, 
\begin{equation}
    \| \rho^{AB}_{G} - \rho^A_{\text{G}} \otimes \rho^B_{\text{G}} \|_1 
    \;\leq\; \tfrac{1}{2} \|M_{\rho_{\text{G}}^{AB}} - M_{\rho_{\text{G}}^A\otimes\rho_{\text{G}}^B}\|_1 .
\end{equation}
Here, \(M_{\rho^A\otimes\rho^B}\) retains the diagonal blocks \(M_A\) and \(M_B\) of \(M_{\rho^{AB}}\) but has zero off-diagonal terms. This corresponds to a separable state with no inter-system correlations. Thus
\[
M_{\rho_{\text{G}}^{AB}} -M_{\rho_{\text{G}}^A\otimes\rho_{\text{G}}^B} 
= \begin{pmatrix} 0 & X \\ -X^{T} & 0 \end{pmatrix}.
\]

The trace norm of the block matrix is determined entirely by the off-diagonal block $X$. We have
\begin{align*}
    \tfrac{1}{2}\,\|M_{\rho^{AB}} - M_{\rho^A \otimes \rho^B}\|_{1}
   & = \frac{1}{2}\left\|\begin{pmatrix}0 & X\\ -X^{ T} & 0\end{pmatrix} \right\|_{1}\\
    &=\frac{\|X\|_{1} + \|X^T\|_{1}}{2}\\
    &=\|X\|_{1}.
\end{align*}
By Lemma~3 of main text, valid for $(k_1,k_2)$-extendible states,
\begin{equation}
X^{ T}X \;\leq\; \tfrac{1}{k_1k_2}\, I ,
\end{equation}
implying that every singular value of $X$ or $X^T$  is at most $1/\sqrt{k_1k_2}$. 
Since there are at most $2\min(n_A,n_B)$ nonzero singular values, we obtain the bound
\begin{equation}
   \tfrac{1}{2}\,\|M_{\rho^{AB}} - M_{\rho^A \otimes \rho^B}\|_{1}
   \;\leq\; \frac{2\,\min(n_A,n_B)}{\sqrt{k_1k_2}}.
\end{equation}

Therefore,
\begin{equation}
    \| \rho^{AB}_{G} - \mathrm{SEP}(A\!:\!B) \|_1 
    \;\leq\; \frac{2\,\min(n_A,n_B)}{\sqrt{k_1 k_2}} 
    = T,
\end{equation}
which establishes the trace-distance bound. 

\subsection*{Continuity bound for mutual information}
The mutual information can be expressed as
\begin{equation}
    I(A;B)_{\rho^{AB}} = S(\rho^A \otimes \rho^B) - S(\rho^{AB}).
\end{equation}
Let
\[
T := \| \rho^A \otimes \rho^B - \rho^{AB} \|_1 .
\]

The Fannes--Audenaert inequality \cite{Audenaert_2007}, states that for two states 
\(\rho,\sigma\) on a \(d\)-dimensional Hilbert space with 
\(\tfrac{1}{2}\|\rho-\sigma\|_1 \le \epsilon\) and \(0 \le \epsilon \le 1\),
\[
|S(\rho)-S(\sigma)| \;\le\; \epsilon \log_2 d + h(\epsilon),
\]
where \(h(\epsilon)=-\epsilon \log_2 \epsilon -(1-\epsilon)\log_2(1-\epsilon)\) is the binary entropy.
Applying this with \(d=2^{n_A+n_B}\) and \(\epsilon=T/2\) gives
\begin{equation}
|S(\rho^A \otimes \rho^B)-S(\rho^{AB})|
   \;\le\; \tfrac{1}{2}(n_A+n_B)T + h(T/2).
\end{equation}
Hence,
\begin{equation}
I(A:B) \;\le\; \tfrac{1}{2}(n_A+n_B)T + h(T/2).
\end{equation}

To ensure the condition \(0 \le \epsilon \le 1\), we define
\begin{equation}
T \;=\; \min\!\left( \frac{2\,\min(n_A,n_B)}{\sqrt{k_1k_2}}, \; 2 \right)
   \;=\; \frac{2}{\sqrt{k_1k_2}}\,\min\!\bigl(n_A,n_B,\sqrt{k_1k_2}\bigr).
\end{equation}
This choice is consistent with the fact that the trace norm distance is at most~2, so we simply set \(T=2\) whenever the bound would exceed this value.

\subsection*{Proof of (ii): Relative entropy of entanglement}
The relative entropy of entanglement is defined as
\begin{equation}
    E_R(\rho^{AB}) := \inf_{\sigma \in \mathrm{SEP}(A:B)} 
    D(\rho^{AB}\,\|\,\sigma),
\end{equation}
where \(D(\rho\|\sigma) = \operatorname{Tr}[\rho(\log \rho - \log \sigma)]\).  

Since the infimum is taken over all separable states, it is upper bounded by choosing \(\sigma = \rho^A \otimes \rho^B\):
\begin{equation}
    E_R(\rho^{AB}_{G}) 
    \;\leq\; D\!\left(\rho^{AB}_{G}\,\big\|\,\rho^A_{\text{G}} \otimes \rho^B_{\text{G}}\right).
\end{equation}
But the right-hand side is exactly the quantum mutual information
\begin{equation}
    I(A:B) := S(\rho_A) + S(\rho_B) - S(\rho^{AB}),
\end{equation}
so that
\begin{equation}
    E_R(\rho^{AB}_{G}) \;\leq\; I(A;B).
\end{equation}
Using the continuity bound derived below, we further obtain
\begin{equation}
    E_R(\rho^{AB}_{G}) \;\leq\; \tfrac{1}{2}(n_A+n_B)T + h(T/2).
\end{equation}

\subsection*{Proof of (iii): Squashed entanglement}
The squashed entanglement is defined as
\begin{equation}
    E_{\mathrm{sq}}(\rho^{AB}) 
    := \tfrac{1}{2} \inf_{\rho^{ABE}} I(A:B|E),
\end{equation}
where the infimum is over extensions \(\rho^{ABE}\) of \(\rho^{AB}\), and  
\begin{equation}
    I(A:B|E) = S(AE)+S(BE)-S(ABE)-S(E)
\end{equation}
is the conditional mutual information.  

By monotonicity of conditional mutual information, 
\begin{equation}
    I(A:B|E) \leq I(A:B).
\end{equation}
Thus
\begin{equation}
    E_{\mathrm{sq}}(\rho^{AB}_{G}) 
    \;\leq\; \tfrac{1}{2} I(A:B).
\end{equation}
Applying the continuity bound below then gives
\begin{equation}
    E_{\mathrm{sq}}(\rho^{AB}_{G}) 
    \;\leq\; \tfrac{1}{4}(n_A+n_B)T + \tfrac{1}{2}h(T/2).
\end{equation}

\section{Construction of Extendible States}

In this section we construct a family of two–mode fermionic Gaussian states. 
We begin by analyzing their extendibility properties and then use them to derive 
lower bounds, with particular attention to the case of two–sided extendibility. 
In this section and the next, we denote a $(k_1,k_2)$–extendible state by $\rho(k_1,k_2)$ 
and its associated CM by $M(k_1,k_2)$.

\subsection{Beam splitter and fermionic pure–loss channel}

Consider the two–mode unitary beam splitter $U_\lambda$, defined by its action on the annihilation operators $\hat{c}_1,\hat{c}_2$ of modes $1$ and $2$:
\begin{align}
    U_{\lambda}^{\dagger}\hat c_1 U_{\lambda} &= \sqrt{\lambda}\,\hat c_1 + \sqrt{1-\lambda}\,\hat c_2,\\
    U_{\lambda}^{\dagger}\hat c_2 U_{\lambda} &= \sqrt{\lambda}\,\hat c_2 - \sqrt{1-\lambda}\,\hat c_1 .
\end{align}
If the second input port is prepared in the vacuum, the first output mode undergoes the fermionic pure–loss channel with transmissivity $\lambda$, denoted $\mathcal{E}_{\lambda}(\cdot)$. For any single–mode input state $\rho=\sum_{m,n=0}^1 \rho_{m,n}\ketbra{m}{n}$, the action is
\begin{align*}
    \mathcal{E}_\lambda(\ketbra{m}{n}) 
    = \sum_{\ell=0}^{\min\{m,n\}}
    \sqrt{\binom{m}{\ell}\binom{n}{\ell}}\,
    (1-\lambda)^{\ell}\,\lambda^{\tfrac{m+n}{2}-\ell}\,
    \ketbra{m-\ell}{n-\ell}.
\end{align*}

\subsection{From EPR pairs to $(k_1,k_2)$–extendible states}

We now use the pure–loss channel to construct a $(k_1,k_2)$–extendible state.  
Starting with an EPR pair, we apply a beam splitter of parameter $\lambda_1$ to subsystem $A$ and another of parameter $\lambda_2$ to subsystem $B$. Tracing through the pure–loss channel action, the EPR state evolves to
\begin{align*}
   \tfrac12\Big(& (1+(1-\lambda_1)(1-\lambda_2))\ketbra{00}{00} 
   + \lambda_2(1-\lambda_1)\ketbra{01}{01} \\
   &+ \lambda_1(1-\lambda_2)\ketbra{10}{10} 
   + \sqrt{\lambda_1\lambda_2}\,(\ket{00}\bra{11}+\ket{11}\bra{00}) 
   + \lambda_1\lambda_2 \ketbra{11}{11}\Big).
\end{align*}

Choosing $\lambda_1=1/k_1$ and $\lambda_2=1/k_2$ yields the CM
\begin{align}
\label{sm:CM-extendible}
    M(k_1,k_2)= 
    \begin{pmatrix} 
        0 & \tfrac{k_1-1}{k_1} & 0 & \tfrac{1}{\sqrt{k_1k_2}} \\
        -\tfrac{k_1-1}{k_1} & 0 & \tfrac{1}{\sqrt{k_1k_2}} & 0 \\
        0 & -\tfrac{1}{\sqrt{k_1k_2}} & 0 & \tfrac{k_2-1}{k_2} \\
        -\tfrac{1}{\sqrt{k_1k_2}} & 0 & -\tfrac{k_2-1}{k_2} & 0
    \end{pmatrix}.
\end{align}
For $k_1=k_2=1$ this reduces to the original EPR CM, which is not extendible.

\subsection{Bona fide condition}

The covariance matrix \(M(k_1,k_2)\) is valid by construction. Indeed, it is obtained from a fermionic Gaussian Bell pair by applying local fermionic pure-loss Gaussian channels with transmissivities \(\lambda_1=1/k_1\) and \(\lambda_2=1/k_2\) to the two subsystems. Fermionic Gaussian channels map valid fermionic Gaussian states to valid fermionic Gaussian states. Hence \(M(k_1,k_2)\) automatically satisfies the bona fide condition.

\subsection{Verification of $(k_1,k_2)$–extendibility}

To show $(k_1,k_2)$–extendibility explicitly, we use Theorem~2. Here
\begin{align*}
    M_A &= \begin{pmatrix} 0 & \tfrac{k_1-1}{k_1} \\[3pt] -\tfrac{k_1-1}{k_1} & 0 \end{pmatrix},\quad
    M_B = \begin{pmatrix} 0 & \tfrac{k_2-1}{k_2} \\[3pt] -\tfrac{k_2-1}{k_2} & 0 \end{pmatrix}, \\
    X &= \begin{pmatrix} 0 & \tfrac{1}{\sqrt{k_1k_2}} \\[3pt] \tfrac{1}{\sqrt{k_1k_2}} & 0 \end{pmatrix}.
\end{align*}
The theorem requires the existence of antisymmetric matrices $Y,Z$ with $I+iY\geq0$ and $I+iZ\geq0$ such that
\begin{align*}
\begin{pmatrix} 
        M_A-Z & 0 & 0 & 0 \\
        0 & M_A + (k_1 - 1)Z & \sqrt{k_1 k_2} X & 0 \\
        0 & -\sqrt{k_1 k_2} X^{T} & M_B + (k_2 - 1)Y & 0 \\
        0 & 0 & 0 & M_B-Y
    \end{pmatrix}
\end{align*}
is a valid CM.  

- For $k_1=k_2=1$, choosing $Y=Z=0$ suffices.  
- For all other cases, one may choose
\begin{align*}
    Z &= \begin{pmatrix} 0 & -1/k_1 \\[3pt] 1/k_1 & 0 \end{pmatrix},\quad
    Y = \begin{pmatrix} 0 & -1/k_2 \\[3pt] 1/k_2 & 0 \end{pmatrix}.
\end{align*}
With these choices the resulting CM corresponds to an EPR pair, and is therefore legitimate. This confirms $(k_1,k_2)$–extendibility.

\medskip

\section{Lower Bounds on the Distance to Fermionic Separable States}
Here we employ the $(k_1,k_2)$-extendible state constructed in the previous section to derive lower bounds on the distance to fermionic separable states. We first present a bound that is tight. We then show that adapting the bosonic strategy of ~\cite{Lami_2019} does not yield a tight lower bound in the fermionic setting; instead, the resulting estimate is weaker by an additional factor of two.

\subsection{Tight Lower Bound from Separable States}

\begin{theorem}
\label{sm:lower-bound}
Let \( \rho \) be a fermionic Gaussian state with CM define above \(M(k_1,k_2)\). Then, the trace distance between \( \rho \) and the set of (possibly non-Gaussian) separable fermionic states satisfies the lower bound
\begin{align*}
       \| \rho - \mathrm{SEP}(A:B) \|_1 \geq \frac{1}{\sqrt{k_1k_2}}.
\end{align*}
\end{theorem}

\begin{proof}
Applying Lemma~\ref{sm:trace-distance-lemma}, we obtain
\begin{align*}
\| M_\rho - M_\sigma \|_\infty \leq \| \rho - \sigma \|_1.
\end{align*}
Minimizing both sides independently over all separable states \( \sigma \), we get
\begin{align*}
\inf_{M_\sigma: \sigma \in \mathrm{SEP}} \|M_\rho - M_\sigma\|_\infty \leq \| \rho - \mathrm{SEP} \|_1.
\end{align*}
Note that the optimal \( \sigma \) minimizing each side may differ. Our goal is to obtain a lower bound on the left-hand side. In fact, we show that the infimum can be computed exactly.

From Theorm~5, any separable (possibly non-Gaussian) fermionic state \( \sigma \) has a correlation matrix of the form
\begin{equation*}
M_\sigma = 
\begin{pmatrix} 
0 & a & 0 & 0 \\
-a & 0 & 0 & 0 \\
0 & 0 & 0 & b \\
0 & 0 & -b & 0
\end{pmatrix},
\end{equation*}
with \( |a|, |b| \leq 1 \). Then,
\begin{align*}
\inf_{M_\sigma:\sigma \in \mathrm{SEP}} \|M_\rho - M_\sigma\|_\infty &= \inf_{|a|\leq 1, |b|\leq 1} \left\|\begin{pmatrix} 
0 & \frac{k_1-1}{k_1} - a & 0 & \frac{1}{\sqrt{k_1k_2}} \\
-\left(\frac{k_1-1}{k_1} - a\right) & 0 & \frac{1}{\sqrt{k_1k_2}} & 0 \\
0 & -\frac{1}{\sqrt{k_1k_2}} & 0 & \frac{k_2-1}{k_2} - b \\
-\frac{1}{\sqrt{k_1k_2}} & 0 & -\left(\frac{k_2-1}{k_2} - b\right) & 0
\end{pmatrix}\right\|_\infty.
\end{align*}

Let us define
\[
\bar{a} := \left( \frac{k_1 - 1}{k_1} - a \right)\sqrt{k_1k_2}, \quad \bar{b} := \left( \frac{k_2 - 1}{k_2} - b \right)\sqrt{k_1k_2}.
\]
Translating the constraints \( |a| \leq 1 \) and \( |b| \leq 1 \) to conditions on \( \bar{a}, \bar{b} \), the problem becomes
\begin{align}
\label{sm:lower-2}
\frac{1}{\sqrt{k_1k_2}} \inf_{\substack{-\sqrt{\frac{k_2}{k_1}} \leq \bar{a} \leq 2\sqrt{k_1k_2}-\sqrt{\frac{k_2}{k_1}}\\ -\sqrt{\frac{k_1}{k_2}} \leq \bar{b} \leq 2\sqrt{k_1k_2}-\sqrt{\frac{k_1}{k_2}}}} \left\|\begin{pmatrix} 
0 & \bar{a} & 0 & 1 \\
-\bar{a} & 0 & 1 & 0 \\
0 & -1 & 0 & \bar{b} \\
-1 & 0 & -\bar{b} & 0
\end{pmatrix}\right\|_\infty = \frac{1}{\sqrt{k_1k_2}}.
\end{align}

The eigenvalues \( \lambda_1, \lambda_2 \) of a real antisymmetric \( 4 \times 4 \) matrix come in purely imaginary conjugate pairs \( \pm i r_1, \pm i r_2 \), with squared moduli given by
\begin{equation}
\lambda_{1,2}^2 = -\frac{\bar{a}^2 + \bar{b}^2}{2} \pm \frac{1}{2}(\bar{a} - \bar{b}) \sqrt{\bar{a}^2 + 2\bar{a}\bar{b} + \bar{b}^2 + 4} - 1.
\end{equation}
The operator norm is the largest singular value, i.e., \( \max\{|r_1|, |r_2|\} \), and since
\[
r_1^2 + r_2^2 = \bar{a}^2 + \bar{b}^2 + 2,
\]
we conclude that
\[
\max\{|r_1|, |r_2|\} \geq 1.
\]
This lower bound is achieved when \( \bar{a} = \bar{b} = 0 \), yielding eigenvalues \( \pm i, \pm i \). Hence, the operator norm is exactly 1, and the infimum in Eq.~\eqref{sm:lower-2} is achieved, completing the proof.
\end{proof}

\begin{proposition}
The trace-distance bound in Theorem~\ref{sm:lower-bound} is tight for two-mode fermionic Gaussian states. Specifically, there exist \((k_1,k_2)\)-extendible Gaussian state \( \rho \), such that:
\begin{align*}
\frac{1}{\sqrt{k_1k_2}} \leq \| \rho - \mathrm{SEP}(A:B) \|_1 \leq \frac{2}{\sqrt{k_1k_2}}.
\end{align*}
\end{proposition}

\begin{proof}
The lower bound follows from Theorem~\ref{sm:lower-bound}, and the upper bound from Theorem~6.
\end{proof}

\subsection{Sharpness of the Mode-Number Dependence}

We now show that the dependence on the number of modes in Theorem~6 is also
sharp.  For simplicity, and because it is already sufficient to prove
optimality of the scaling, we consider the balanced case
\[
        k_1=k_2=k .
\]
The construction below gives a family of \((k,k)\)-extendible fermionic
Gaussian states on \(r\) modes in \(A\) and \(r\) modes in \(B\), with
\(r=\min(n_A,n_B)\), whose distance from the fermionic separable set is lower
bounded by a constant multiple of \(r/k\).  Since in the balanced case
\[
        \frac{\min(n_A,n_B)}{\sqrt{k_1k_2}}
        =
        \frac{r}{k},
\]
this proves that the mode-number dependence in Theorem~6 cannot be improved,
up to universal constants, before the trivial saturation of the trace norm.

Let \(a_j\) and \(b_j\), \(j=1,\ldots,r\), denote fermionic annihilation
operators for the \(A\)- and \(B\)-modes.  For each pair \((a_j,b_j)\), define
new fermionic modes
\begin{equation}
        c_j=\frac{a_j+b_j}{\sqrt{2}},
        \qquad
        d_j=\frac{a_j-b_j}{\sqrt{2}} .
\end{equation}
Fix
\[
        p=\frac{1}{k}.
\]
On the two modes \((c_j,d_j)\), consider the Gaussian state
\begin{equation}
        \tau_p
        =
        \Big((1-p)\ketbra{0}_{c_j}+p\ketbra{1}_{c_j}\Big)
        \otimes
        \ketbra{0}_{d_j}.
\end{equation}
Equivalently, in the original \(A:B\) mode basis,
\begin{equation}
        \tau_p
        =
        (1-p)\ketbra{00}
        +
        p\ketbra{\psi_+},
        \qquad
        \ket{\psi_+}
        =
        \frac{\ket{10}+\ket{01}}{\sqrt{2}} .
\end{equation}
The state is fermionic Gaussian because it is diagonal in the fermionic normal
modes \(c_j,d_j\).  Its covariance matrix, in the Majorana ordering
\((\gamma_1^A,\gamma_2^A,\gamma_1^B,\gamma_2^B)\), is
\begin{equation}
        M_p
        =
        \begin{pmatrix}
        0 & -(1-p) & 0 & -p \\
        1-p & 0 & p & 0 \\
        0 & -p & 0 & -(1-p) \\
        p & 0 & 1-p & 0
        \end{pmatrix}.
        \label{sm:eq:Mp-sharpness}
\end{equation}
Thus the off-diagonal covariance block has singular values equal to \(p\).

We first verify that \(\tau_p\) is \((k,k)\)-extendible.  In the notation of
Theorem~2, write
\[
        M_A=M_B=
        \begin{pmatrix}
        0 & -(1-p)\\
        1-p & 0
        \end{pmatrix},
        \qquad
        X=
        \begin{pmatrix}
        0 & -p\\
        p & 0
        \end{pmatrix}.
\]
Choose
\[
        \Delta_A=\Delta_B=
        \begin{pmatrix}
        0 & -1\\
        1 & 0
        \end{pmatrix},
\]
which is the covariance matrix of a one-mode vacuum state and hence is a valid
fermionic covariance matrix.  Since \(p=1/k\), we have
\begin{equation}
        kM_A-(k-1)\Delta_A=0,
        \qquad
        kM_B-(k-1)\Delta_B=0,
\end{equation}
and
\[
        \sqrt{k^2}\,X
        =
        kX
        =
        \begin{pmatrix}
        0 & -1\\
        1 & 0
        \end{pmatrix}.
\]
Therefore the effective covariance matrix appearing in Theorem~2 is
\begin{equation}
        \begin{pmatrix}
        0 & kX\\
        -kX^T & 0
        \end{pmatrix},
\end{equation}
which is a valid covariance matrix, since its singular values are all equal to
one.  Hence \(\tau_p\) is \((k,k)\)-extendible.

Now define the \(r\)-pair Gaussian state
\begin{equation}
        \rho_{r,k}^{AB}
        :=
        \tau_p^{\otimes r},
        \qquad p=\frac{1}{k}.
\end{equation}
It is Gaussian, with covariance matrix
\[
        M_{\rho_{r,k}}
        =
        \bigoplus_{j=1}^{r} M_p .
\]
It is also \((k,k)\)-extendible, because the tensor product of the
\((k,k)\)-extensions of the individual copies gives a \((k,k)\)-extension of
the tensor-product state.  If \(n_A\) or \(n_B\) is larger than \(r\), we may
tensor with additional uncorrelated local vacuum modes; this does not change
the distance estimate below.

We now prove a lower bound on the distance from \(\rho_{r,k}^{AB}\) to the
fermionic separable set.  Let \(\ket{\Omega}\) denote the global vacuum of the
\(a_j,b_j\) modes, and define the one-particle states
\begin{equation}
        \ket{A_j}=a_j^\dagger\ket{\Omega},
        \qquad
        \ket{B_j}=b_j^\dagger\ket{\Omega},
        \qquad j=1,\ldots,r.
\end{equation}
Let \(P_1\) be the projection onto the one-particle subspace
\[
        \mathcal{H}_1
        =
        \operatorname{span}\{\ket{A_j},\ket{B_j}:j=1,\ldots,r\}.
\]
Since \(\rho_{r,k}^{AB}\) is a product of the states \(\tau_p\), its restriction
to the one-particle sector is
\begin{equation}
        P_1\rho_{r,k}^{AB}P_1
        =
        p(1-p)^{r-1}
        \sum_{j=1}^{r}
        \ketbra{\psi_j},
        \qquad
        \ket{\psi_j}
        =
        \frac{\ket{A_j}+\ket{B_j}}{\sqrt{2}} .
        \label{sm:eq:one-particle-block}
\end{equation}
Therefore the off-diagonal block of \(P_1\rho_{r,k}^{AB}P_1\), between
\(\operatorname{span}\{\ket{A_j}\}_{j=1}^{r}\) and
\(\operatorname{span}\{\ket{B_j}\}_{j=1}^{r}\), is
\begin{equation}
        C_\rho
        =
        \frac{p(1-p)^{r-1}}{2}\, I_r .
        \label{sm:eq:C-rho-sharpness}
\end{equation}

On the other hand, let \(\sigma^{AB}\in \mathrm{SEP}(A:B)\) be any fermionic
separable state, in the sense used throughout this work: it is a convex
combination of physical product states,
\begin{equation}
        \sigma^{AB}
        =
        \sum_\ell q_\ell\,\sigma_\ell^A\otimes\sigma_\ell^B,
        \qquad
        [\sigma_\ell^A,P^A]=0,\quad [\sigma_\ell^B,P^B]=0 .
\end{equation}
For every \(j,\ell\), the matrix element
\(\bra{A_j}\sigma^{AB}\ket{B_\ell}\) vanishes.  Indeed, in a product component,
\[
        \bra{A_j}
        \bigl(\sigma_m^A\otimes\sigma_m^B\bigr)
        \ket{B_\ell}
        =
        \bra{A_j}\sigma_m^A\ket{\Omega_A}\,
        \bra{\Omega_B}\sigma_m^B\ket{B_\ell}.
\]
The first factor connects opposite local parity sectors of \(A\), and the
second factor connects opposite local parity sectors of \(B\).  Since
\(\sigma_m^A\) and \(\sigma_m^B\) commute with the corresponding local parity
operators, both such parity-changing matrix elements vanish.  Hence the
off-diagonal block of \(P_1\sigma^{AB}P_1\) between the \(A\)-one-particle and
\(B\)-one-particle subspaces is zero.

It follows that the corresponding off-diagonal block of
\(P_1(\rho_{r,k}^{AB}-\sigma^{AB})P_1\) is exactly \(C_\rho\).  For a Hermitian
block matrix
\[
        H=
        \begin{pmatrix}
        H_A & C\\
        C^\dagger & H_B
        \end{pmatrix},
\]
one has
\[
        \|H\|_1
        \geq
        \left\|
        \begin{pmatrix}
        0 & C\\
        C^\dagger & 0
        \end{pmatrix}
        \right\|_1
        =
        2\|C\|_1 .
\]
The first inequality follows by applying the contraction
\(H\mapsto (H-UHU)/2\), where \(U=I_A\oplus(-I_B)\), and the equality follows
from the singular-value decomposition of \(C\).

Using also \(\|P_1HP_1\|_1\leq \|H\|_1\), we obtain, for every separable
\(\sigma^{AB}\),
\begin{align}
        \|\rho_{r,k}^{AB}-\sigma^{AB}\|_1
        &\geq
        \|P_1(\rho_{r,k}^{AB}-\sigma^{AB})P_1\|_1  \\
        &\geq
        2\|C_\rho\|_1 \\
        &=
        2\left\|
        \frac{p(1-p)^{r-1}}{2}I_r
        \right\|_1 \\
        &=
        r\,p(1-p)^{r-1}.
\end{align}
Taking the infimum over all separable \(\sigma^{AB}\), and substituting
\(p=1/k\), gives
\begin{equation}
        \|\rho_{r,k}^{AB}-\mathrm{SEP}(A:B)\|_1
        \geq
        \frac{r}{k}\left(1-\frac{1}{k}\right)^{r-1}.
        \label{sm:eq:sharpness-lower-rk}
\end{equation}
In the nontrivial regime \(r/k\leq 1/2\), Bernoulli's inequality gives
\[
        \left(1-\frac{1}{k}\right)^{r-1}
        \geq
        1-\frac{r-1}{k}
        \geq
        \frac{1}{2},
\]
and therefore
\begin{equation}
        \|\rho_{r,k}^{AB}-\mathrm{SEP}(A:B)\|_1
        \geq
        \frac{r}{2k}.
        \label{sm:eq:sharpness-lower-final}
\end{equation}

Combining this lower bound with the upper bound of Theorem~6,
\[
        \|\rho_{r,k}^{AB}-\mathrm{SEP}(A:B)\|_1
        \leq
        \frac{2r}{k},
\]
we conclude that the dependence on the number of modes is sharp, up to
universal numerical constants, in the regime before the trace norm saturates.
Equivalently, since \(r=\min(n_A,n_B)\) and \(k=\sqrt{k_1k_2}\) in the balanced
case \(k_1=k_2=k\), the scaling
\[
        \frac{\min(n_A,n_B)}{\sqrt{k_1k_2}}
\]
in Theorem~6 cannot be improved in general.

\subsection{Adapting the Bosonic Strategy}
To estimate \( \left\| \rho(k_1,k_2) - \text{SEP} \right\|_1 \), we first recall that for any traceless operator \( X \) it holds that \( \|X\|_1 \geq 2\|X\|_\infty \). This directly yields the lower bound
\begin{equation}
\left\| \rho(k_1,k_2) - \text{SEP} \right\|_1 \geq 2 \left\| \rho(k_1,k_2) - \text{SEP} \right\|_\infty .
\end{equation}

From Lemma \ref{sm:convex-set} we know that, for every bipartite pure state \( |\Psi\rangle \) with maximal Schmidt coefficient \( \lambda_1(\Psi) \), the inequality
\(
\langle \Psi | \sigma | \Psi \rangle \leq \lambda_1(\Psi)
\)
holds for all separable states \( \sigma \). Hence,
\begin{equation}
\left\| \rho(k_1,k_2) - \text{SEP} \right\|_\infty 
= \inf_{\sigma \in \text{SEP}} \left\| \rho(k_1,k_2) - \sigma \right\|_\infty ,
\end{equation}
\begin{equation}
= \inf_{\sigma \in \text{SEP}} \sup_{|\Psi\rangle} \left| \langle \Psi | \left( \rho(k_1,k_2) - \sigma \right) | \Psi \rangle \right| ,
\end{equation}
\begin{equation}
\geq \sup_{|\Psi\rangle} \left( \bra{\Psi} \rho(k_1,k_2) | \Psi \rangle - \lambda_1(\Psi) \right).
\end{equation}

We now specialize to the case where \( |\Psi\rangle \) is chosen as the EPR pair 
\(
\frac{1}{\sqrt{2}}(\ket{00}+\ket{11}),
\)
which satisfies \( \lambda_1(\Psi) = \tfrac{1}{2} \). In this case,
\(
\langle \Psi | \rho(k_1,k_2) | \Psi \rangle = \tr\!\left( \rho(k_1,k_2)\rho(1,1) \right).
\)

For two fermionic Gaussian states \(\rho(k_1,k_2)\) and \(\rho(1,1)\) with CMs \(M(k_1,k_2)\) and \(M(1,1)\), respectively, one has the identity
\begin{equation}
    \big|\tr(\rho(k_1,k_2)\rho(1,1))\big| 
    = \sqrt{\det\!\left(\tfrac{1}{2}\big(M(k_1,k_2)M(1,1)-I\big)\right)} .
\end{equation}

Applying this relation, we obtain
\begin{align}
     \sqrt{ \det\!\left( \tfrac{1}{2}\big(M(k_1,k_2)M(1,1) - I\big) \right) }
     = \frac{1}{4}\left(\frac{2k_1k_2-k_1-k_2+2}{k_1k_2}+\frac{2}{\sqrt{k_1k_2}}\right) -\frac{1}{2}.
\end{align}

It is straightforward to verify that this bound is always weaker than the one derived in Proposition \ref{sm:lower-bound}. For the symmetric case \(k_1=k_2=k\), the expression simplifies to
\begin{equation}
    \sqrt{ \det\!\left( \tfrac{1}{2}\big(M(k,k)M(1,1) - I\big) \right) }
    = \frac{1}{2}\left(1+\frac{1}{k^2}\right),
\end{equation}
which leads to
\begin{equation}
\left\| \rho(k) - \text{SEP} \right\|_1 \geq \frac{1}{k^2}.
\end{equation}

Recall that the bound from Theorem \ref{sm:lower-bound} in the same setting is \(\tfrac{1}{k}\). Therefore,
\begin{equation}
    \frac{1}{k} \geq \frac{1}{k^2}, \quad \text{for all } k \geq 1.
\end{equation}
We emphasize that this inequality reflects more than a difference in constant prefactors: it highlights a change in scaling behavior, with the first bound decaying linearly in \(k\), while the second decays quadratically.

\medskip

\section{Independence of Extendibility from Opposite Subsystems}  

This section establishes that extendibility constraints can behave in a strikingly asymmetric manner. 
A natural question is the following: given extendibility from one subsystem, what can be inferred about extendibility from the other? Surprisingly, the two sides turn out to be completely independent. An extension from one side imposes no constraint on the other, except in the trivial case where one side is infinitely extendible. This is counterintuitive, as one might expect that extendibility from one side should push the state closer to separable, and that, due to the scarcity of entanglement, this would in turn enforce extendibility from the other side. Instead, the extendibility conditions for each side are decoupled.  
Throughout this section, we use the shorthand ``1-extendible'' interchangeably with ``unextendible.''  

\begin{figure}[h!]
\centering
\begin{tikzpicture}[scale=0.7, x=2.4cm,y=-2.4cm,
  arrInt/.style={->, very thick},
  arrInf/.style={<->, very thick},
  cont/.style={dotted, thick},
  every node/.style={font=\small}
]
\node (A11) at (0,0) {\textbf{$(1,1)$}};
\node (A21) at (1,0) {\textbf{$(2,1)$}};
\node (A31) at (2,0) {\textbf{$(3,1)$}};

\node (A12) at (0,1) {\textbf{$(1,2)$}};
\node (A22) at (1,1) {\textbf{$(2,2)$}};
\node (A32) at (2,1) {\textbf{$(3,2)$}};

\node (A13) at (0,2) {\textbf{$(1,3)$}};
\node (A23) at (1,2) {\textbf{$(2,3)$}};
\node (A33) at (2,2) {\textbf{$(3,3)$}};

\foreach \m in {1,2,3} {
  \draw[arrInt] (A2\m) -- (A1\m);
  \draw[arrInt] (A3\m) -- (A2\m);
}
\foreach \n in {1,2,3} {
  \draw[arrInt] (A\n2) -- (A\n1);
  \draw[arrInt] (A\n3) -- (A\n2);
}

\node (R1) at (3.2,0) {\textbf{$(\infty,1)$}};
\node (R2) at (3.2,1) {\textbf{$(\infty,2)$}};
\node (R3) at (3.2,2) {\textbf{$(\infty,3)$}};

\draw[cont] (A31) -- (R1);
\draw[cont] (A32) -- (R2);
\draw[cont] (A33) -- (R3);

\draw[arrInf] (R1) -- (R2);
\draw[arrInf] (R2) -- (R3);

\node (C1) at (0,3.2) {\textbf{$(1,\infty)$}};
\node (C2) at (1,3.2) {\textbf{$(2,\infty)$}};
\node (C3) at (2,3.2) {\textbf{$(3,\infty)$}};

\draw[cont] (A13) -- (C1);
\draw[cont] (A23) -- (C2);
\draw[cont] (A33) -- (C3);

\draw[arrInf] (C1) -- (C2);
\draw[arrInf] (C2) -- (C3);

\node (B) at (3.2,3.2) {\textbf{$(\infty,\infty)$}};
\draw[arrInf] (R3) -- (B);
\draw[arrInf] (C3) -- (B);

\def\xleft{-0.5}   
\def\xright{2.8}   
\def\xinner{3.7}   
\def\ybottom{-0.5} 
\def\yinner{3.6}   
\def\ytop{2.7}     

\draw[red, very thick, dashed]
  (\xleft,\ytop) -- (\xright,\ytop) -- (\xright,\ybottom)
  -- (\xinner,\ybottom) -- (\xinner,\yinner) -- (\xleft,\yinner) -- cycle;


\end{tikzpicture}
\caption{
Hierarchy of $(k_1,k_2)$-extendible sets. The diagram depicts the lattice of extendibility constraints, where the horizontal and vertical directions correspond to increasing $k_1$ and $k_2$, respectively. Arrows indicate strict inclusion of sets: larger $k_i$ yield strictly smaller families (as proved by Lemma~\ref{sm:Hierarch}). Notably, extendibility from one side imposes no constraint on the other, so the two directions are independent, except in the limiting case $k_1\to\infty$ or $k_2\to\infty$, where the sets coincide with the separable states, highlighted by the red dashed boundary.}

    \label{sm:fig:geometry_schematic}
\end{figure}

To illustrate this independence, consider the following family of two-mode fermionic Gaussian states.  For \(k_1,k_2\geq 1\), define
\begin{equation}
\label{sm:CM-extendible-epsilon}
    M_\epsilon(k_1,k_2)
    =
    \begin{pmatrix}
    \left(1-\frac{\epsilon}{k_1}\right)i\sigma_y
    &
    \frac{\sqrt{2\epsilon-\epsilon^2}}{\sqrt{k_1k_2}}\,\sigma_x
    \\
    -\frac{\sqrt{2\epsilon-\epsilon^2}}{\sqrt{k_1k_2}}\,\sigma_x
    &
    \left(1-\frac{\epsilon}{k_2}\right)i\sigma_y
    \end{pmatrix},
\end{equation}
where \(0<\epsilon<1\), and \(\sigma_x,\sigma_y\) are Pauli matrices.  We choose \(\epsilon\) sufficiently small so that
\begin{equation}
\label{sm:epsilon-condition-hierarch}
0<\epsilon<
\min\left\{
1,\,
\frac{2k_2}{k_1k_2+k_2-k_1},\,
\frac{2k_1}{k_1k_2+k_1-k_2}
\right\}.
\end{equation}

\begin{lemma}
The state with CM \(M_\epsilon(k_1,k_2)\) as above is \((k_1,k_2)\)-extendible but not \((k_1+1,1)\)- or \((1,k_2+1)\)-extendible.
\label{sm:Hierarch}
\end{lemma}

\begin{proof}
We first show that \(M_\epsilon(k_1,k_2)\) is \((k_1,k_2)\)-extendible.  In the notation of Theorem~2, its blocks are
\begin{align}
    M_A
    &=
    \left(1-\frac{\epsilon}{k_1}\right)i\sigma_y,
    \qquad
    M_B
    =
    \left(1-\frac{\epsilon}{k_2}\right)i\sigma_y,
    \\
    X
    &=
    \frac{\sqrt{2\epsilon-\epsilon^2}}{\sqrt{k_1k_2}}\,\sigma_x .
\end{align}
Choose
\begin{equation}
    \Delta_A=i\sigma_y,
    \qquad
    \Delta_B=i\sigma_y .
\end{equation}
Then
\begin{align}
    k_1M_A-(k_1-1)\Delta_A
    &=
    (1-\epsilon)i\sigma_y,
    \\
    k_2M_B-(k_2-1)\Delta_B
    &=
    (1-\epsilon)i\sigma_y,
    \\
    \sqrt{k_1k_2}\,X
    &=
    \sqrt{2\epsilon-\epsilon^2}\,\sigma_x .
\end{align}
Thus the effective CM appearing in Theorem~2 is
\begin{equation}
    \begin{pmatrix}
    (1-\epsilon)i\sigma_y
    &
    \sqrt{2\epsilon-\epsilon^2}\,\sigma_x
    \\
    -\sqrt{2\epsilon-\epsilon^2}\,\sigma_x
    &
    (1-\epsilon)i\sigma_y
    \end{pmatrix}.
\end{equation}
Using
\[
(i\sigma_y)^2=-I,\qquad
\sigma_x^2=I,\qquad
(i\sigma_y)\sigma_x+\sigma_x(i\sigma_y)=0,
\]
and
\[
(1-\epsilon)^2+(2\epsilon-\epsilon^2)=1,
\]
one checks that the square of this matrix is \(-I\).  Hence it is a valid pure fermionic CM.  Therefore, by Theorem~2, \(M_\epsilon(k_1,k_2)\) is \((k_1,k_2)\)-extendible.

We now show that it is not \((k_1+1,1)\)-extendible.  Suppose, for contradiction, that a \((k_1+1,1)\)-extension exists.  Applying Theorem~2 with parameters \((k_1+1,1)\), the corresponding effective CM would contain the \(B\)-block
\[
    M_B
    =
    \left(1-\frac{\epsilon}{k_2}\right)i\sigma_y
\]
and the off-diagonal block
\[
    \sqrt{k_1+1}\,X
    =
    \sqrt{\frac{(k_1+1)(2\epsilon-\epsilon^2)}{k_1k_2}}\,\sigma_x .
\]
For any valid fermionic CM, Lemma~\ref{sm:column-sum} implies that the sum of the squared entries in each column is at most one.  Applying this to either column of the \(B\)-subsystem in the effective CM gives the necessary condition
\begin{equation}
    \left(1-\frac{\epsilon}{k_2}\right)^2
    +
    \frac{(k_1+1)(2\epsilon-\epsilon^2)}{k_1k_2}
    \leq 1 .
\end{equation}
For \(\epsilon>0\), this inequality is equivalent to
\begin{equation}
    \epsilon
    \geq
    \frac{2k_2}{k_1k_2+k_2-k_1}.
\end{equation}
This contradicts the choice of \(\epsilon\) in Eq.~\eqref{sm:epsilon-condition-hierarch}.  Hence the state is not \((k_1+1,1)\)-extendible.

The proof that the state is not \((1,k_2+1)\)-extendible is analogous.  If such an extension existed, applying Lemma~\ref{sm:column-sum} to either column of the \(A\)-subsystem in the corresponding effective CM would give
\begin{equation}
    \left(1-\frac{\epsilon}{k_1}\right)^2
    +
    \frac{(k_2+1)(2\epsilon-\epsilon^2)}{k_1k_2}
    \leq 1 .
\end{equation}
For \(\epsilon>0\), this is equivalent to
\begin{equation}
    \epsilon
    \geq
    \frac{2k_1}{k_1k_2+k_1-k_2},
\end{equation}
again contradicting Eq.~\eqref{sm:epsilon-condition-hierarch}.  Therefore the state is not \((1,k_2+1)\)-extendible.
\end{proof}

\begin{corollary}
For every finite \(k\), there exists a fermionic Gaussian state that is \((1,k)\)-extendible but not \((2,1)\)-extendible.  Similarly, there exists a fermionic Gaussian state that is \((k,1)\)-extendible but not \((1,2)\)-extendible.
\end{corollary}

The family \(M_\epsilon(k_1,k_2)\) therefore shows that extendibility constraints from opposite subsystems are independent: being extendible to \(k_1\) copies on one side and \(k_2\) copies on the other does not force the next nontrivial extension from either opposite side.

\begin{theorem}
For every \(\varepsilon > 0\), there exists a fermionic Gaussian state \(\rho_\varepsilon\) such that \(\|\rho_\varepsilon - \mathrm{SEP}\|_{1} \leq \varepsilon\), but which is not extendible from either side. In particular, it admits only a 1-extension.  
\end{theorem}

\begin{proof}
The proof is constructive. Consider the family of CMs
\begin{equation}
    C_\varepsilon = 
    \begin{pmatrix}
    0 & \sqrt{1-(\frac{\varepsilon}{2})^2} & \frac{\varepsilon}{2} & 0 \\
    -\sqrt{1-(\frac{\varepsilon}{2}) ^2} & 0 & 0 & \frac{\varepsilon}{2}  \\
    -\frac{\varepsilon}{2}  & 0 & 0 & \sqrt{1-(\frac{\varepsilon}{2}) ^2} \\
    0 & -\frac{\varepsilon}{2}  & -\sqrt{1-(\frac{\varepsilon}{2}) ^2} & 0
    \end{pmatrix}.
\end{equation}

One readily checks that \(C_\epsilon\) defines a valid fermionic CM. In the limit \(\epsilon \to 0\), the corresponding state approaches separability. However, for any fixed \(\epsilon > 0\), the state remains unextendible.  

The first nontrivial extendibility conditions arise at the \((1,2)\) and \((2,1)\) levels. Since any higher-order extension would imply one of these, it suffices to rule out both \((1,2)\)- and \((2,1)\)-extendibility to establish that the state is only \((1,1)\)-extendible.  

Decomposing \(C_\epsilon\) into blocks, we write
\[
M_A = M_B = 
\begin{pmatrix}
0 & \sqrt{1 - (\frac{\varepsilon}{2})^2} \\
- \sqrt{1 - (\frac{\varepsilon}{2})^2} & 0
\end{pmatrix}, \quad
X = 
\begin{pmatrix}
\frac{\varepsilon}{2} & 0 \\
0 & \frac{\varepsilon}{2}
\end{pmatrix}.
\]

By Theorem~2, for \((k_1, k_2) = (1,2)\), the state would be extendible if and only if there exist real antisymmetric matrices \(Y = -Y^T\) and \(Z = -Z^T\) such that the following \(8 \times 8\) block matrix defines a valid fermionic CM (see Eq.~\eqref{sm:eq:new_CM}):
\begin{equation}
\begin{pmatrix} 
M_A - Z & 0 & 0 & 0 \\
0 & M_A & \sqrt{2}X & 0 \\
0 & -\sqrt{2}X^{T} & M_B + Y & 0 \\
0 & 0 & 0 & M_B - Y
\end{pmatrix}.
\end{equation}

Applying Lemma~\ref{sm:column-sum} to the second row shows that this matrix fails to satisfy the CM condition for any choice of \(Y\) and \(Z\). Hence the state is not \((1,2)\)-extendible. By symmetry between subsystems \(A\) and \(B\), it is also not \((2,1)\)-extendible.  

We conclude that \(\rho_\epsilon\) is only \((1,1)\)-extendible, despite being arbitrarily close to separable.
\end{proof}

\section{Numerical illustration: disordered Kitaev chain}

To complement the general covariance-matrix characterization developed in the main text, we include a simple
many-body example based on a disordered quadratic fermionic model.  The purpose of this appendix is not to
analyze the phase diagram of the model, but rather to illustrate concretely how our extendibility
criterion can be applied to reduced states arising from a standard many-body system.

Specifically, we ask the following question: given the ground state of a disordered free-fermion Hamiltonian,
how often is the reduced state on two neighboring sites \((k_1,k_2)\)-extendible?  Since extendibility provides a natural way to probe how local correlations are modified, we apply this to understand how localized the correlations are as disorder increases.  For the example
considered below, we find that the probability of extendibility increases with disorder strength, consistent
with the intuition that stronger disorder suppresses inter-site correlations and tends to localize them.

As a concrete testbed, we consider the standard disordered Kitaev chain, namely a one-dimensional spinless
\(p\)-wave superconductor.  At the symmetric point \(t=\Delta=1\), the Hamiltonian on a ring of length \(N\)
is
\begin{equation}
\label{sm:eq:disordered_kitaev_chain}
H \;=\; \sum_{j=1}^{N}\Big[-(c_j^\dagger c_{j+1} + c_{j+1}^\dagger c_j)
\;+\; (c_j c_{j+1} + c_{j+1}^\dagger c_j^\dagger)\Big]
\;-\;\sum_{j=1}^{N}\mu_j\Big(c_j^\dagger c_j-\tfrac12\Big),
\end{equation}
with periodic boundary conditions (\(c_{N+1}\equiv c_1\)).  The onsite chemical potentials are sampled
independently from the uniform distribution
\[
\mu_j \sim \mathrm{Unif}[-W,W],
\]
where \(W\ge 0\) denotes the disorder strength.

Because the model is quadratic, its many-body ground state is a fermionic Gaussian state and is therefore
completely determined by its Majorana covariance matrix.  Writing the Majorana operators as
\(\gamma=(\gamma_1,\dots,\gamma_{2N})^{\mathsf T}\), any quadratic Hamiltonian can be expressed as
\begin{equation}
H \;=\; \frac{i}{4}\,\gamma^{\mathsf T} A\,\gamma,
\end{equation}
where \(A\) is a real antisymmetric \(2N\times 2N\) matrix.  For a state \(\rho\), the corresponding Majorana
covariance matrix is defined by
\begin{equation}
\Gamma_{pq} \;:=\; \frac{i}{2}\,\mathrm{Tr}\!\big(\rho\,[\gamma_p,\gamma_q]\big)
\;=\; \mathrm{Tr}\!\big(\rho\, i\gamma_p\gamma_q\big)\qquad (p\neq q),\qquad \Gamma_{pp}=0.
\end{equation}
For a pure Gaussian ground state, \(\Gamma\) is real antisymmetric and satisfies \(\Gamma^2=-I\).

In the numerics, we compute the ground-state covariance matrix \(\Gamma\) of the quadratic Hamiltonian from the Hermitian single-particle generator \(K:=iA\) via
\begin{equation}
\Gamma \;=\; i\,\mathrm{sgn}(K),
\end{equation}
where \(\mathrm{sgn}(K)\) is obtained by diagonalizing \(K\) and replacing each eigenvalue by its sign.

\begin{figure}[t]
  \centering
  \includegraphics[width=0.5\linewidth]{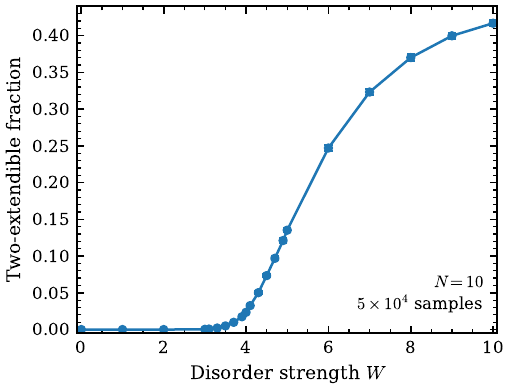}
  \caption{
    \textbf{Disordered Kitaev chain: two-site \((k_1,k_2)\)-extendibility versus disorder.}
    For the disordered Kitaev chain in Eq.~\eqref{sm:eq:disordered_kitaev_chain} with periodic boundary conditions and
    i.i.d.\ chemical-potential disorder \(\mu_j\sim\mathrm{Unif}[-W,W]\), we compute the Gaussian ground-state
    covariance matrix, reduce it to two adjacent sites, and test two-sided \((k_1,k_2)\)-extendibility of the
    resulting two-site Gaussian state using the covariance-matrix criterion developed in the main text.
    The plotted quantity is the empirical fraction \(N_{\mathrm{ext}}(W)/N_{\mathrm{samp}}\) of disorder
    realizations for which the reduced state is \((1,2)\)-extendible. Error bars denote \(95\%\) Wilson score
    intervals. Parameters are \(N=10\), \(t=\Delta=1\), \((k_1,k_2)=(1,2)\), and
    \(N_{\mathrm{samp}}=50\,000\) independent disorder realizations for each value of \(W\).
    }
  \label{sm:fig:disordered_kitaev_two_site_extendibility}
\end{figure}

For each value of the disorder strength \(W\), we sample \(N_{\mathrm{samp}}\) independent realizations and
define
\begin{equation}
p(W)\;:=\;\Pr_{\{\mu_j\}}\!\Big[\rho^{(12)}_G \text{ is } (k_1,k_2)\text{-extendible}\Big]
\;\approx\;
\frac{N_{\mathrm{ext}}(W)}{N_{\mathrm{samp}}},
\end{equation}
where \(N_{\mathrm{ext}}(W)\) is the number of sampled disorder realizations, out of \(N_{\mathrm{samp}}\),
for which the reduced two-site state is \((k_1,k_2)\)-extendible.

The figure shows the empirical probability \(p(W)\) as a function of \(W\), for fixed \((k_1,k_2)\) and fixed
global system size \(N\).  We stress that extendibility is tested only for the reduced two-site state; the
role of the full system size \(N\) is merely to define the parent many-body Gaussian ground state from which
this local reduction is obtained.  Since the quantity under consideration is local, one expects the curve to
become essentially insensitive to \(N\) once \(N\) is sufficiently large.

The figure shows a smooth disorder-driven increase in the empirical extendible fraction. For \(W\leq 2\),
no extendible realization is observed among the \(50\,000\) samples at each plotted disorder strength,
whereas a small but nonzero fraction is already resolved at \(W=3\). The extendible fraction then increases
steadily with \(W\). Because the calculation is performed for a fixed finite chain and a finite Monte Carlo
sample, the apparent onset should be interpreted as a crossover in an ensemble probability rather than as
evidence for a sharp thermodynamic phase transition.

This trend is consistent with the interpretation of extendibility as a constraint on the strength and
structure of bipartite correlations. Increasing disorder tends to suppress effective cross-correlations
across a fixed local cut, thereby making the reduced two-site state progressively more likely to satisfy the
\((k_1,k_2)\)-extendibility conditions. In this way, the calculation provides a concrete illustration of how
the covariance-matrix extendibility criterion developed in the main text can be applied to local reduced
states in disordered fermionic Gaussian systems.

\end{document}